\documentclass{birkjour}
\usepackage{graphics}
\usepackage[T1]{fontenc}
\usepackage{dsfont}
\usepackage{amsopn, amssymb,amsthm}
\usepackage{bbm}
\usepackage[german,english]{babel} 
\usepackage{url}   
\newcommand{\mfS}{\mathfrak{S}}
\newcommand{\wh}{\widehat}

\DeclareMathOperator{\tr}{\mathrm{tr}}

\newcommand{\Psy}{P_{\sy}}
\newcommand{\Pasy}{P_{\asy}}
\newcommand{\ex}{\mathrm{ex}}
\newcommand{\inte}{\mathrm{int}}
\newcommand{\new}{\mathrm{new}}
\newcommand{\sy}{\mathrm{sy}}
\newcommand{\asy}{\mathrm{asy}}

\newcommand{\eb}{\eqref}
\newcommand{\Pbse}{Q^{\bot}}
\newcommand{\Pb}{Q}
\newcommand{\ord}{\mathrm{ord}}
\newcommand{\bal}{\bsy{\alpha}}
\newcommand{\bbe}{\bsy{\beta}}
\newcommand{\bw}{\boldsymbol{w}}
\newcommand{\bz}{\boldsymbol{z}}

\newcommand{\by}{\boldsymbol{y}}
\newcommand{\bv}{\boldsymbol{v}}
\newcommand{\va}{\boldsymbol{a}}

\newcommand{\bc}{\boldsymbol{c}}
\newcommand{\bx}{\boldsymbol{x}}

\newcommand{\Ind}{\mathrm{Ind}}
\newcommand{\Mat}{\mathrm{Mat}}

\DeclareMathOperator{\im}{\mathrm{Im}}
\newcommand{\ran}{\mathrm{ran}}

\newcommand{\ud}{\mathrm{d}}
\newcommand{\ui}{\mathrm{i}}
\newcommand{\ue}{\mathrm{e}}

\newcommand{\rz}{{\mathbb R}}
\newcommand{\nz}{{\mathbb N}}
\newcommand{\kz}{{\mathbb C}}

\newcommand{\eins}{\mathds{1}}
\newcommand{\be}{\begin{equation}}
\newcommand{\ee}{\end{equation}}

\newcommand{\lk}{\left(}
\newcommand{\rk}{\right)}
\newcommand{\lgk}{\left\{}
\newcommand{\rgk}{\right\}}
\newcommand{\mf}{\mathfrak}
\newcommand{\mc}{\mathcal}
\newcommand{\ba}{\begin{aligned}}
\newcommand{\ea}{\end{aligned}}

\newcommand{\wt}{\widetilde}

\newcommand{\lin}{\mathrm{lin}}
\newcommand{\bsy}{\boldsymbol}
\newcommand{\bma}{\left(
\begin{array}
}
\newcommand{\ema}{\end{array}
\right)}
 \newtheorem{theorem}{Theorem}[section]
 \newtheorem{cor}[theorem]{Corollary}
 \newtheorem{lemma}[theorem]{Lemma}
 \newtheorem{prop}[theorem]{Proposition}
 \theoremstyle{definition}
 \newtheorem{defn}[theorem]{Definition}
 \theoremstyle{remark}
 \newtheorem{rem}[theorem]{Remark}
 \newtheorem*{exa}{Example}
 \numberwithin{equation}{section}
\begin{document}
%
%
%
%
%
%
%
%
%
%
\title[Zero modes and an index theorem on metric graphs]{Zero modes of quantum 
graph Laplacians and an index theorem}

\author{Jens Bolte, Sebastian Egger}

\address{Department of Mathematics, Royal Holloway, University of London, 
Egham, TW20 0EX, UK}

\email{jens.bolte@rhul.ac.uk, sebastian.egger@rhul.ac.uk}


\author{Frank Steiner}

\address{Institut f{\"u}r Theoretische Physik, Universit{\"a}t Ulm \\ 
Albert-Einstein-Allee 11, 89081 Ulm, Germany \\  and \\ 
Universit\'e de Lyon, Observatoire de Lyon, Centre de Recherche 
Astrophysique de Lyon, CNRS UMR 5574: Universit\'e Lyon 1 and 
\'Ecole Normale Sup\'erieure de Lyon, 9, avenue 
Charles Andr\'e, F-69230 Saint-Genis-Laval, France}

\email{frank.steiner@uni-ulm.de}




\date{today}

\begin{abstract}
We study zero modes of Laplacians on compact and non-compact metric graphs with 
general self-adjoint vertex conditions. In the first part of the paper the number 
of zero modes is expressed in terms of the trace of a unitary matrix $\mf{S}$ 
that encodes the vertex conditions imposed on functions in the domain of the Laplacian. 
In the second part a Dirac operator is defined whose square is related to the
Laplacian. In order to accommodate Laplacians with negative eigenvalues it is
necessary to define the Dirac operator on a suitable Kre\u{\i}n space. We demonstrate
that an arbitrary, self-adjoint quantum graph Laplacian admits a factorisation into
momentum-like operators in a Kre\u{\i}n-space setting. As a consequence, we establish
an index theorem for the associated Dirac operator and prove that the zero-mode 
contribution in the trace formula for the Laplacian can be expressed in terms of the 
index of the Dirac operator. 
\end{abstract}

\maketitle

\section{Introduction}
Index theorems play an important role in differential topology. They relate topological 
data of a vector bundle over a manifold to spectral data of a differential operator 
acting on sections in that bundle. They hence reveal the extent to which topological 
information about a space is encoded in spectra of certain differential 
operators. In physics such information is very useful when the vector bundles are 
associated with principal bundles in a gauge theory. This is, e.g., the case in the 
Aharonov-Casher effect related to the zero modes of a magnetic Schr\"odinger operator.
Central to the formulation of an index theorem is a suitable Dirac operator. Sometimes,
however, one is interested in the zero modes of a certain non-negative operator
(which could be a Laplacian or a Schr\"odinger operator) and realises that these
are best revealed in terms of a Dirac operator squaring to the operator in question.
In such a case the question is whether a given operator admits a factorisation in
the form $H=AA^\ast$, where $A^\ast$ is the adjoint of an operator $A$.

In a fairly general setting, a Dirac operator is a self-adjoint operator in a 
Hilbert space $\mf{H}=\mf{H}_1\oplus\mf{H}_2$ that possesses a representation
\be
\label{420}
D=\bma{cc}
M_1 & d^{\ast}\\
d & -M_2
\ema
\ee
with respect to the decomposition of $\mf{H}$ into the sub-Hilbert spaces $\mf{H}_{1/2}$.
Its (dense) domain $\mc{D}:=\mc{D}_{d}\oplus\mc{D}_{d^{\ast}}$ splits correspondingly, 
with $\quad\mc{D}_{d}\subseteq\mf{H}_1$ and $\mc{D}_{d^{\ast}}\subseteq\mf{H}_2$.
Furthermore,  $d^{\ast}:\mc{D}_{d^{\ast}}\rightarrow\mf{H}_1$ is the adjoint of
$d:\mc{D}_{d}\rightarrow\mf{H}_2$, and the operators 
$M_{1/2}:\mc{D}_{M_{1/2}}\rightarrow\mf{H}_{1/2}$ are self-adjoint, see 
\cite[p. 151]{Thaller:1992} for more details. Assuming that $dM_1=M_2d$ and
$d^{\ast}M_2=M_1d^{\ast}$, the squared Dirac operator is diagonal,
\be
\label{8111}
D^2=\bma{cc}
d^{\ast}d+M_1^2 & 0 \\
0& dd^{\ast}+M_2^2
\ema .
\ee
Often, $d$ will be a differential in a de Rham complex. Then the operators on 
the diagonal of \eqref{8111} are Laplacians plus, possibly, lower order terms. 
A simple example would be $\mf{H}_1=\mf{H}_2$ and $M_{1}=-M_{2}=m\eins$. The 
standard Dirac-Hamiltonian in quantum mechanics (cf. \cite[Section 1]{Thaller:1992}) 
is of this type, where $d$ is the exterior derivative on $\rz^n$ and $m$ is the 
mass of the particle. Below we shall treat Dirac operators on metric graphs where, 
for simplicity, we put $M_1=0$, $M_2=0$. The analytical index of the Dirac operator 
$D$ in \eqref{420} is defined as (cf. \cite[p. 158]{Thaller:1992})
\be
\label{1040}
\ba
\Ind (D)&:=\dim\ker d-\dim\ker d^{\ast}\\ 
        &=\dim\ker d^\ast d -\dim\ker dd^\ast ,
\ea
\ee
where $\dim \ker d^{\ast}<\infty$ and $\dim \ker d<\infty$ is assumed.

Quantum graphs are models of differential operators on (metric) graphs. They 
were introduced in \cite{Kottos:1997,KottosSmilansky:1998} as one-dimensional 
models with a complex structure. It was found that their spectral properties are
the same as in classically chaotic quantum systems in that the spectral correlations 
follow the predictions of random matrix theory. These observations initiated many 
further numerical and analytical investigations (see, e.g., 
\cite{AGA:2008,Berkolaiko:2013}).  Standard Dirac operators on intervals were 
already considered in \cite{Alonso:1997}, and on compact metric graphs in 
\cite{Bulla:1990,Bolte:2003}. A trace formula for the spectral density of standard 
Dirac operators on metric graphs was derived in \cite{Bolte:2003}.

The squared Dirac operator \eqref{8111} is obviously non-negative, and the same is 
true for the operators $d^\ast d$ and $dd^\ast$ on the diagonal. Quantum graph 
Laplacians, however, may possess finitely many negative eigenvalues. Therefore, not 
every self-adjoint quantum graph Laplacian can be factorised in the form $d^\ast d$. 
This can only be achieved if there is no negative spectrum. The factorisation of 
non-negative Laplacians of certain types was given in \cite{Gaveau:1991,Fulling:2007,Post:2009} 
and associated index theorems were proven. In \cite{Fulling:2007} self-adjoint 
realisations of quantum graph Laplacians with non-Robin vertex conditions were 
considered and, using two independent methods, the indices of the associated Dirac 
operators were calculated. One method utilises the well-known heat-kernel technique, 
whereas the other method is based on an edge-wise construction of zero modes. The 
cases covered in \cite{Fulling:2007} do not exhaust all non-negative quantum graph 
Laplacians. Self-adjoint realisations with Robin vertex conditions, such that the 
operators are still non-negative, were treated in \cite{Post:2009} and an index
theorem was proved. In that context Dirac operators 
\be
\label{5001}
D=\bma{cc}
0 & d^\ast\\
d& 0
\ema
\ee
respecting the local structure at the vertices, i.e., the topology of the metric 
graph, were constructed in such a way that the negative Laplacian is given by
 $-\Delta=d^\ast d$. This required an 
appropriate de Rham complex involving boundary values of functions at the vertices.  

In the case of non-Robin vertex conditions the index of the associated Dirac operator 
was found to be of the form  \cite{Fulling:2007}
\be
\label{1021a}
\Ind\lk D\rk=\frac{1}{2}\tr\mf{S}_{\Delta}.
\ee
Here $\mf{S}_{\Delta}$ is the vertex scattering matrix or, briefly, the $\mf{S}$-matrix, 
encoding the vertex conditions imposed on functions in the domain of $-\Delta$. An 
alternative form of the index that was devised in \cite{Fulling:2007} involves the 
extent to which the vertex conditions are of the Dirichlet-type, and this form of the 
index theorem was generalised in \cite{Post:2009}.

Zero modes of quantum graph Laplacians not only play a role in the index \eqref{1040},
but also contribute to the trace formula for the Laplacian on a compact
graph. In \cite{BE:2008} it was shown that
\be
\label{TF}
\ba
\sum_{k_n^2\geq 0}h(k_n)
 &= \mathcal{L}\hat h(0) + \gamma h(0) -\frac{1}{4\pi}\int_{-\infty}^{\infty}
    {h(k)}\,\frac{\im\tr\mf{S}_\Delta (k)}{k} \ dk \\
 &\quad +\sum_{p}\left[(\hat h *\hat A_p)(l_p) +
    (\hat h *\hat{\bar A}_p)(l_p)\right],
\ea
\ee
where $\mathcal{L}$ is the total length of the graph and the sum on the right-hand 
side extends over all periodic orbits $p$ on the graph; $l_p$ is the length of $p$
and $A_p(k)$ is an associated amplitude function. The sum on the left-hand side is
over all Laplace eigenvalues $k_n^2\geq 0$, and $h$ is a suitable test function with
Fourier transform $\hat h$. Crucially, the constant $\gamma\in\rz$ in the term
summarising the contribution of zero modes could only be determined in terms of
the spectral and an algebraic multiplicity of the zero eigenvalue. For Laplacians
with non-Robin vertex conditions it was found that
\be
\gamma = \frac{1}{4}\tr\mf{S}_{\Delta} = \frac{1}{2}\Ind\lk D\rk.
\ee
Hence, for this class of Laplacians an important observation on the role of the
graph topology in the trace formula was made.

When a Laplacian has negative eigenvalues a factorisation $-\Delta=d^\ast d$ as described 
above is impossible. This has to do with the Hilbert space structure involved, 
implying that an operator of the form $AA^{\ast}$ is always non-negative. 
Therefore, in any attempt at achieving a factorisation of an arbitrary self-adjoint 
realisation of a quantum graph Laplacian one would have to give up the Hilbert space 
structure. It is our goal in this paper to provide an approach in which a Dirac operator 
is defined in a Kre\u{\i}n space setting, where the positive-definite inner product is 
replaced by a non-degenerate hermitian form. This constitutes a minimal deviation from 
the conventional setting described above such that negative Laplace eigenvalues can 
still be incorporated. In this context we shall eventually prove an index theorem that 
is a natural generalisation of the results in \cite{Fulling:2007,Post:2009} and will 
be presented in the form \eqref{1021a}.

The paper is organised as follows: We recall the basic definitions and properties 
of metric graphs and self-adjoint realisations of Laplacians in Section~\ref{3010}. 
In Section~\ref{3111} we define spectral and algebraic multiplicities of zero modes
and identify a condition under which the two concepts of an algebraic multiplicity
coincide. Zero modes of Laplacians are then characterised in Section~\ref{111}.
The trace of the $\mf{S}$-matrix is computed in Section~\ref{111x} and, for compact
graphs, used to express the zero-mode contribution in the trace formula 
\eqref{TF}. In Section \ref{8403} we present an explicit example of a Dirac operator 
on an interval in the sense of \cite{Post:2009}. This will suggest how to generalise 
the setting in Sections \ref{402az} and \ref{402}. There the associated Kre\u{\i}n space 
structure and Dirac operators are derived and the factorisation of the Laplacians by 
means of the Dirac operators is constructed. Finally, an index theorem is proved in 
Section~\ref{402}. 
\section{Metric graphs and differential Laplacians}
\label{3010}
We recall (see, e.g., \cite{KostrykinSchrader:1999,KostrykinSchrader:2006b}) that 
a metric graph $\Gamma=\left(\mc{V},\mc{E}_{\ex},\mc{E}_{\inte},\partial,\mc{I}\right)$ 
consists of a finite set of vertices $\mc{V}$, a finite set of internal edges 
$\mc{E}_{\inte}$, a finite set of external edges $\mc{E}_{\ex}$ and a map $\partial$ 
assigning to every internal edge $e\in \mc{E}_{\inte}$ an ordered pair of vertices, 
$\partial(e):= \left(v_1,v_2\right)$, and to every external edge $e\in\mc{E}_{\ex}$ 
a single vertex, $\partial(e)=v$. The vertices $\partial(e)$ are the edge ends of $e$.
Furthermore, to every edge $e\in\mc{E}:=\mc{E}_{\inte}\cup\mc{E}_{\ex}$ an interval 
$I_e=\left(0,l_e\right)$ is assigned, such that $0<l_e<\infty$ if $e\in\mc{E}_{\inte}$ 
and $l_e=\infty$ if $e\in\mc{E}_{\ex}$. The set of these intervals is denoted as
$\mc{I}$. A metric graph is said to be compact if $\mc{E}_{\ex}=\emptyset$, and 
non-compact otherwise.

We say that an edge $e$ is adjacent to a vertex $v$ if $v\in\partial(e)$, and denote 
this as $e\sim v$. The degree $d(v)$ of a vertex is the number of edges adjacent to $v$. 
We emphasise that it is allowed for two vertices to be connected by more than one 
internal edge, and that the edge ends of an internal edge may coincide, 
$\partial(e)=(v,v)$, thus producing a loop. A loop counts double in the degree of that
vertex.

A function $\psi$ on $\Gamma$ is a collection of functions on the edges, 
$\psi=\{\psi_e\}_{e\in\mc{E}}$ with $\psi_e:I_e\to\kz$. Accordingly, function spaces 
are defined as direct sums, e.g., 
\be
\label{26}
L^2(\Gamma):=\bigoplus_{e\in\mc{E}}L^2 \left(0,l_e\right).
\ee
Similar constructions apply to Sobolev spaces, $H^m(\Gamma)$, and spaces of smooth 
functions, $C^\infty(\Gamma)$, etc. The notation \eqref{26} in particular implies that 
\emph{no} conditions are imposed at the vertices. The space \eqref{26}, equipped with 
its standard inner product, is a Hilbert space. A function $\psi\in H^1(\Gamma)$ has 
well-defined boundary values at edge ends, which we use to define 
\begin{equation}
\label{12}
\underline{\psi} := \left(\left\{\psi_{e}(0)\right\}_{e\in{\mc{E}_{\inte}}},
\left\{\psi_{e}\left(l_{e}\right)\right\}_{e\in{\mc{E}_{\inte}}},\left\{\psi_{e}(0)
\right\}_{e\in{\mc{E}_{\ex}}}\right)^T\in\kz^E,
\end{equation}
where $E:=2\left|\mc{E}_{\inte}\right|+\left|\mc{E}_{\ex}\right|$. 

Apart from \eqref{26} we shall also need $p$-form spaces, $p\in\lgk 0,1\rgk$, as in 
\cite{Post:2009}. The $0$-form space for a metric graph $\Gamma$ is $L^2(\Gamma)$ 
and the corresponding $1$-form space is given by $L^2(\Gamma)\oplus\kz^E$. 
Equipped with their standard inner products the $p$-form spaces are Hilbert spaces.

If $\mc{S}\subset\kz^E$ is a subspace,  we denote  the corresponding 
orthogonal projection as $P_{\mc{S}}$, and the projection to the orthogonal 
complement $\mc{S}^\perp$ of $\mc{S}$ as $P_{\mc{S}}^\perp$. We also define a 
linear map $I:\kz^E\to\kz^E$ in terms of its matrix representation as
\be
\label{30}
I:=
\left(\begin{array}{ccc}
\eins_{\mc{E}_{\inte}} & 0 & 0\\
0 & -\eins_{\mc{E}_{\inte}} & 0 \\
0  & 0 & \eins_{\mc{E}_\ex}
\end{array}\right).
\ee
We denote the (weak) derivative on $H^1(\Gamma)$ as $(\cdot)'$, such that 
$\psi'=\{\psi_e'\}_{e\in\mc{E}}$ with
\begin{equation}
\label{203a}
\psi_e' = \frac{\ud\psi_e}{\ud x}.
\end{equation}
We also sometimes use a {\it momentum operator}
\be
p_{\Gamma}:=-\ui \left(\cdot\right)'.
\ee
As a differential operator the Laplacian is defined to be 
\be
\label{3001}
-\Delta\psi:={p_\Gamma}^2\psi=-\psi''
\ee
for $\psi\in C^\infty(\Gamma)$. The self-adjoint realisations of the Laplacian
were given in \cite{KostrykinSchrader:1999,Kuchment:2004,KostrykinSchrader:2006b}.
The parametrisation given in \cite{Kuchment:2004} is as follows.
\begin{theorem}
\label{3000}
Let $\mc{D}_{\Delta}\subset L^2(\Gamma)$ be a domain of a self-adjoint realisation of 
the Laplacian $-\Delta$ with core $C_0^\infty(\Gamma)$. Then this domain can be uniquely 
characterised as
\be
\label{3003a}
\mc{D}_{\Delta}=\lgk\psi\in H^2(\Gamma); 
\quad \lk P+L\rk\underline{\psi}+P^{\bot}I\underline{\psi'}=0\rgk,
\ee
where $P:\kz^E\to\kz^E$ is an orthogonal projection and $L:\kz^E\to\kz^E$ is 
self-adjoint, satisfying $P^{\bot}LP^{\bot}=L$. 
\end{theorem}
We also denote such a self-adjoint Laplacian $\lk-\Delta,\mc{D}_{\Delta}\rk$ as 
$-\Delta_{P,L}$.
\begin{defn}
\label{3013}
\begin{enumerate}
\item The vertex scattering matrix or, briefly, the $\mf{S}$-matrix, for 
$-\Delta_{P,L}$ is defined as
\be
\label{3005}
\mf{S}_{P,L}(k):=-P-\lk L+\ui k P^{\bot}\rk^{-1}\lk L-\ui kP^{\bot}\rk, \quad k\in\kz.
\ee
The Laplacian is said to possess $k$-independent vertex conditions if 
$\mf{S}_{P,L}$ is independent of $k$. Otherwise they are called $k$-dependent.
\item The Laplacian is said to be local with respect to $\Gamma$ if the 
$\mf{S}$-matrix allows the decomposition 
\be
\label{3011}
\mf{S}_{P,L}(k)=\oplus_{v\in\mc{V}}\mf{S}_{P_v,L_v}(k), \quad 
\mbox{for all}\quad k\in\kz, 
\ee 
where 
\be
\label{3016a}
\mf{S}_{P_v,L_v}(k):\kz^{d(v)}\rightarrow\kz^{d(v)}, \quad \mbox{for all} 
\quad v\in\mc{V} \quad \mbox{and} \quad k\in\kz. 
\ee
\item The Laplacian is said to be strictly local with respect to $\Gamma$ if 
every $\mf{S}_{P_v,L_v}(k)$ is irreducible, i.e., there exists no refinement of
\eqref{3016a}.  
\end{enumerate}
\end{defn}
\begin{rem}
\label{c3}
\begin{itemize}
\item The $\mf{S}$-matrix is meromorphic in $\kz$ and unitary for $k\in\rz$. 
The pole structure is described in \cite{BE:2008}.
\item For local Laplacians there exist decompositions of the maps $L$, $P$ and 
$P^{\bot}$ that correspond to \eqref{3011}. Hence, according to the vertex conditions
imposed by \eqref{3003a}, only boundary values of functions on edges that are 
adjacent to the same vertex are related to each other.
\item Non-Robin vertex conditions correspond to the $k$-independent case, where 
$L=0$, whereas Robin vertex conditions correspond to the $k$-dependent case, with 
$L\neq0$.
\item As $\ran L\subset\ker P=\ran P^\perp$, the projectors $P$ and $P_{\ran L}$
are orthogonal with respect to each other so that $Q:=P+P_{\ran L}$ is also a projector. 
\end{itemize}
\end{rem}
\begin{lemma}[\cite{BE:2008}]
\label{9004}
Let $\mf{S}_{P,L}$ be given and set $Q:=P+P_{\ran L}$. Then
\be
\label{9006}
\ba
\mf{S}_\infty &:=\lim\limits_{k\rightarrow\infty}\mf{S}_{P,L}(k) = P^{\bot}-P,\\
\mfS_0  &:= \lim\limits_{k\rightarrow 0}\mf{S}_{P,L}(k) = Q^{\bot}-Q.
\ea
\ee
\end{lemma}
In \cite{BE:2008} the second line of \eqref{9006} was not explicitly given in 
this form, but it follows immediately from eq.\ (3.5) in \cite{BE:2008}. 
\section{Multiplicities of the Laplace-eigenvalue zero}
\label{3111}
The principal tool to characterise Laplace eigenvalues is a secular equation
\cite{KostrykinSchrader:2006b,BE:2008}. In order to set this up one requires the 
following quantity in addition to the $\mf{S}$-matrix.
\begin{defn}
\label{32}
Let $\bsy{l}=(l_1,\dots,l_{|\mc{E}_{\inte}|})^T$ be the vector of finite edge lengths 
and let $\ue^{\ui k\bsy{l}}$ be a diagonal matrix with diagonal entries $\ue^{\ui kl_e}$,
$e\in\mc{E}_{\inte}$. For $k\in\kz$ define the matrix
\be
\label{3.1}
T(k;\bsy{l}):=
\left(\begin{array}{ccc}
T_{\inte}(k;\bsy{l}) & 0\\
0 & 0_{\mc{E}_{\ex}} 
\end{array}\right),
\ee
acting on $\kz^E=\kz^{2|\mc{E}_\inte|}\oplus\kz^{|\mc{E}_\ex|}$, where
\be
T_{\inte}(k;\bsy{l})=
\left(\begin{array}{cc}
0 & \ue^{\ui k \bsy{l}}  \\
\ue^{\ui k \bsy{l}} & 0\\ 
\end{array}\right).
\ee
\end{defn}
\begin{rem}
\label{33}
When $k\in\rz$, the matrix $T_{\inte}(k;\bsy{l})$ is unitary. The same is true for 
$T(k;\bsy{l})$ iff $\mc{E}_{\ex}=\emptyset$, i.e., for a compact graph.
\end{rem}
The secular function is introduced on the basis of the following result
\cite{KostrykinSchrader:2006}.
\begin{prop}
Let $k\in\kz\setminus[\ui\sigma(L)\cup\{0\}]$, then $k^2$ is an eigenvalue of
$-\Delta_{P,L}$ with spectral multiplicity $g$, iff $1$ is an eigenvalue of the 
matrix $\mf{S}_{P,L}(k)T(k;\bsy{l})$ with geometric multiplicity $g$.
\end{prop}
Hence, the secular function
\be
\label{35}
F(k):=\det\lk\eins_{\kz^{E}}-\mf{S}_{P,L}(k)T(k;\bsy{l})\rk,
\quad k\in\kz\setminus[\ui\sigma(L)\cup\{0\}],
\ee
has a zero at $k$, iff $k^2$ is an eigenvalue of $-\Delta_{P,L}$. 
Setting $l_{\min}:=\min\{l_e;\quad e\in\mc{E}_{\inte}\}$ and
\begin{equation}
\label{70}
\lambda^+_{\min} :=
  \begin{cases} 
  \min\{\lambda\in\sigma(L),\quad\lambda>0\}, & \text{if}\ \sigma(L)
  \cap\rz_+\ne\emptyset, \\ \infty, & \text{else},
  \end{cases}
\end{equation}
it was shown in \cite{BE:2008} that on a compact graph $\Gamma$ the 
spectral multiplicity of a Laplace eigenvalue $k^2\neq 0$ coincides with the 
order of the zero $k$ of the secular function $F$, if one assumes that 
$l_{\min}>2/\lambda^+_{\min}$.

In the following we shall extend the analysis to $k=0$ and to non-compact 
$\Gamma$. We shall also clarify the role of the additional assumption. For this 
we require the following notions.
\begin{defn}
\label{36}
\begin{itemize}
\item A \emph{zero mode} of the Laplacian $-\Delta_{P,L}$ is a non-trivial element of
$\ker\Delta_{P,L}\subset\mc{D}_\Delta$. The \emph{spectral multiplicity} of the eigenvalue
zero is
\be
\label{37}
g_0:=\dim\ker\Delta_{P,L}.
\ee
\item The \emph{algebraic multiplicity} $N$ of the eigenvalue zero \emph{in the sense of} 
\cite{Fulling:2007} is the order of the zero of $F$ at $k=0$,
\be
\label{38}
N:=\left.\ord [F]\right|_{k=0}. 
\ee
\item The \emph{algebraic multiplicity} $\wt{N}$ of the eigenvalue zero \emph{in the 
sense of} \cite{Kurasov:2010} is
\be
\label{38a}
\wt{N}:=\dim\ker\lk\eins_{\kz^{E}}-\mfS_{0}\mf{J}_{\mc{E}}\rk,
\ee
where 
\be
\label{42}
\mf{J}_{\mc{E}}:=
\left(\begin{array}{cc}
J_{\mc{E}_{\inte}} & 0  \\
0 & 0_{\mc{E}_{\ex}}
\end{array}\right)
,\quad
J_{\mc{E}_{\inte}}:=
\left(
\begin{array}{cc}
0 & \eins_{\mc{E}_{\inte}}  \\
\eins_{\mc{E}_{\inte}} & 0 
\end{array}
\right).
\ee
\end{itemize}
\end{defn}
As $T(0;\bsy{l})=\mf{J}_{\mc{E}}$, comparing $F(0)$ with \eqref{38a} one may be 
tempted to assume that $N=\wt{N}$. The following example, however, shows that 
this need not be the case. 
\begin{exa}
\label{40}
Consider a graph consisting of one edge and two vertices, to which the interval 
$I=\left[0,l\right]$ is assigned. Choose $P=0$ and $L=\lambda\eins_{\kz^2}$, 
$\lambda\in\rz\setminus\{0\}$, implying Robin conditions at the vertices. Using 
that any zero mode is of the form $\psi_0(x)=\alpha x+\beta$ and imposing the 
vertex conditions gives the spectral multiplicity of the eigenvalue zero as
\be
\label{45}
g_0=
\begin{cases}
1, &  \text{if} \ l=\frac{2}{\lambda},\\
0, & \text{else}.
\end{cases}
\ee
Furthermore, $\mfS_{0}=-\eins_{\kz^2}$ leads to the secular function
\be
\label{43}
F(k)=1-\lk\frac{\lambda-\ui k}{\lambda+\ui k}\rk^2\ue^{2\ui kl}.
\ee
Hence,
\be
\label{41}
N=\left.\ord [F]\right|_{k=0}=\begin{cases}
3, & \text{if}\ l=\frac{2}{\lambda},\\
1, & \text{else}.
\end{cases}
\ee
However,
\be
\wt{N}=\dim\ker\left(\begin{array}{cc}
1 & 1 \\
1 & 1
\end{array}
\right)=1.
\ee
Under the additional assumption $l>2/\lambda$ required in \cite{BE:2008}
(see also the paragraph below \eqref{70}), one finds that $N=\wt{N}$. This 
is not true when $l=2/\lambda$. However, the quantity $\gamma$ in \eqref{TF}, 
which was determined in \cite{BE:2008} to be $\gamma =g_0 -N/2$, satisfies
\be
\label{gammaS0}
\gamma =\frac{1}{4}\tr\mf{S}_0 =-\frac{1}{2}
\ee
in {\it all} cases of this example. 
\end{exa}
In the following we shall prove that the first equality in \eqref{gammaS0} holds
for all self-adjoint Laplacians on compact graphs. Before, we shall derive a 
sufficient criterion for $N=\wt{N}$ to hold, which in the example above simply 
excludes the case $l=2/\lambda$. For this we require certain subspaces of $\kz^E$.
\begin{defn}
\label{46}
Define the following subspaces of the space $\kz^E$ of boundary values:
\be
\label{47}
M_{\sy}:=\left\{
\lk\bc,\bc,\bsy{0}_{\mc{E}_{\ex}}\rk^T
\in\kz^{E}, \quad \bc\in \kz^{|\mc{E}_{\inte}|}\right\},
\ee
\be
\label{48}
M_{\asy}:=\left\{
\lk\bc,-\bc,\bsy{0}_{\mc{E}_{\ex}}\rk^T
\in\kz^{E},\quad \bc\in \kz^{|\mc{E}_{\inte}|}\right\},
\ee
\be
\label{49}
M_0:=\left\{ \lk\bsy{0}_{\mc{E}_{\inte}},\bsy{0}_{\mc{E}_{\inte}},\bc\rk^T\in\kz^{E}, 
\quad \bc\in\kz^{|\mc{E}_{\ex}|}\right\}.
\ee
We also set
\be
\label{50}
M:=M_{\sy}\oplus M_{\asy}.
\ee
\end{defn}
The projectors onto $M_{\sy}$, $M_{\asy}$ and $M_0$ are denoted
as $P_{\sy}$, $P_{\asy}$ and $P_0$, respectively. In matrix form the
first two read
\be
\label{Psyays}
P_{\sy}=\frac{1}{2}\left(\begin{array}{ccc} 
\eins_{\mc{E}_{\inte}} & \eins & 0\\
\eins & \eins_{\mc{E}_{\inte}} & 0\\
 0 & 0 & 0_{\mc{E}_{\ex}}
\end{array}\right), \quad
P_{\asy}=\frac{1}{2}\left(\begin{array}{ccc} 
\eins_{\mc{E}_{\inte}} & -\eins & 0\\
-\eins & \eins_{\mc{E}_{\inte}} & 0\\
 0 & 0 & 0_{\mc{E}_{\ex}}
\end{array}\right),
\ee
such that
\be 
\label{JP}
P_{\sy}-P_{\asy}=\mf{J}_{\mc{E}}=T(0;\bsy{l}), 
\ee
see \eqref{42}.
 
We need to consider the eigenvalue problem of the matrix
\be
 U(k):=\mf{S}_{P,L}(k) T(k;\bsy{l}),\quad k\in\rz. 
\ee
\begin{lemma}
\label{52}
Let $k\in\rz$. In the eigenvalue problem
\be
\label{51}
U(k)\bw=\lambda\bw,
\ee
where, in general, $\bw\in\kz^E$ and $\lambda\in\kz$, the following
equivalence holds:
\be
\label{53}
\bw\in M \Leftrightarrow |\lambda|=1.
\ee
Furthermore, $M_0$ is the eigenspace corresponding to the eigenvalue
$\lambda=0$.
\end{lemma}
\begin{proof}
From the unitarity of $\mf{S}_{P,L}(k)$ for $k\in\rz$ and from \eqref{3.1}
one infers that
\be
\label{56}
\left\|\mf{S}_{P,L}(k) T(k;\bsy{l})\bw\right\|_{\kz^{E}}
\leq\left\|\bw\right\|_{\kz^{E}},
\ee
implying $|\lambda|\leq 1$. Moreover, a strict inequality holds iff
$P_0\bw\neq\bsy{0}$.

It is obvious from \eqref{3.1} that $T(k;\bsy{l})\bw=\bsy{0}$ iff $\bw\in M_0$,
implying the last claim.
\end{proof}
In general, when $\Gamma$ is non-compact and hence $\mc{E}_{\ex}\neq\emptyset$, 
the matrix $U(k)=\mf{S}_{P,L}(k) T(k;\bsy{l})$ is neither unitary nor can it be 
diagonalised. Its eigenvectors do not span the entire space $\kz^{E}$ and eigenvectors 
corresponding to different eigenvalues are not orthogonal with respect to each other. 
The following observation, however, will turn out to be sufficient.
\begin{lemma}
\label{58}
Let $k\in\rz$ and $\lambda(k)$ be an eigenvalue of $U(k)$. If $|\lambda(k)|=1$ the
eigenvalue is semi-simple.
\end{lemma}
\begin{proof}
Due to Lemma \ref{52} any eigenvector $\bw$ corresponding to $\lambda(k)$ with 
$|\lambda(k)|=1$ is in $M$. We now modify $U(k)$ on the orthogonal complement 
$M^\perp=M_0$ to this subspace by setting 
\be
\wt{U}(k):=\mf{S}_{P,L}(k)\wt{T}(k;\bsy{l}),
\ee
where
\be
\label{59}
\wt{T}(k;\bsy{l}):=
\left(\begin{array}{cc}
T_{\inte}(k;\bsy{l}) & 0\\
0 & \eins_{\mc{E}_{\ex}} 
\end{array}\right).  
\ee
When $k\in\rz$ the matrix $\wt{U}(k)$ is unitary. Therefore, its eigenvalues 
$\wt{\lambda}(k)$ satisfy $|\wt{\lambda}(k)|=1$, and there exists an orthonormal 
basis of eigenvectors $\lgk\bx_1,\ldots,\bx_{E}\rgk$ for $\kz^E$. We arrange this 
basis in such a way that $\bx_1,\dots,\bx_F\in M$ (where $1\leq F\leq 2|\mc{E}_\inte|$), 
and set
\be
\label{60}
\wt{M}:=\lin\lgk\bx_l; \quad 1\leq l\leq F\rgk\subset M.
\ee
As $U(k)$ and $\wt{U}(k)$ coincide on $M$, the eigenvectors 
$\bx_l\in\wt{M}$ are also eigenvectors of $U(k)$ with the same eigenvalues. In fact, 
by Lemma~\ref{52} they exhaust all eigenvectors of $U(k)$ with eigenvalues satisfying 
$|\lambda(k)|=1$. 

By construction, $\wt{U}(k)$ leaves $\wt{M}$ and $\wt{M}^{\bot}$ invariant, 
\be
\label{62}
\wt{U}(k)=P_{\wt{M}}\wt{U}(k) P_{\wt{M}} + P_{\wt{M}}^{\bot}\wt{U}(k)P_{\wt{M}}^{\bot},
\ee 
where $P_{\wt{M}}$ and $P_{\wt{M}}^{\bot}$ are the projectors onto $\wt{M}$ and 
$\wt{M}^{\bot}$, respectively. This then implies that
\be
\label{63}
U(k)\bsy{x}=\wt{U}(k)\bsy{x}\in\wt{M},\quad \mbox{for every} \quad 
\bsy{x}\in\wt{M}\subset M,
\ee
so that $U(k)$ leaves $\wt{M}$ invariant, too. Now let $\bx\in\wt{M}^{\bot}$.
Since $P_0\bx\in M_0=M^\perp\subset\wt{M}^{\bot}$ and $TP_0\bx=\bsy{0}$, one obtains
\be
\label{64}
U(k)\bsy{x}=U(k)\lk\bsy{x}-P_0\bsy{x}\rk=\wt{U}(k)\lk\bsy{x}-P_0\bsy{x}\rk
\in\wt{M}^{\bot},
\ee
so that $U(k)$ also leaves $\wt{M}^\perp$ invariant. Hence,
\be
\label{66}
U(k)=P_{\wt{M}}U(k)P_{\wt{M}}+P_{\wt{M}}^{\bot}U(k)P_{\wt{M}}^{\bot}.
\ee 
We therefore conclude that the Jordan blocks in the Jordan normal forms of 
$P_{\wt{M}}(k)\wt{U}(k)P_{\wt{M}}(k)$ and $P_{\wt{M}}(k)U(k)P_{\wt{M}}(k)$ coincide. 
As the eigenvectors of $U(k)$ corresponding to eigenvalues with $|\lambda(k)|=1$
are in $\wt{M}$, and $\wt{U}(k)$ can be diagonalised, these eigenvalues are 
semi-simple.
\end{proof}
Together with Theorem~1 in \cite[p. 402]{Lancaster:1985}, Lemma~\ref{58} ensures that 
any eigenvalue (-function) $\lambda(k)$ of $U(k)$ such that there exists $k_0$
with $|\lambda(k_0)|=1$ is (real) differentiable in a neighbourhood of $k_0$.
Furthermore, 
\be
\label{lambexp}
\lambda(k) = \lambda(k_0) + \sum_{n=l}^\infty a_n(k-k_0)^{n/l},
\ee
where $a_n\in\kz$ and $l$ is an integer not exceeding the multiplicity of
the eigenvalue $\lambda(k_0)$. There also exists a corresponding eigenvalue
(-function) $\bx(k)$ with an expansion
\be
\label{xexp}
\bx(k) = \sum_{n=0}^\infty \bx_n(k-k_0)^{n/l},
\ee
that converges for $0<|k-k_0|<\delta$ with some $\delta>0$. Moreover,
$\bx_n\in\ker(U(k_0)-\lambda(k_0))$ for $n=0,\dots,l-1$.

The above results enable us to prove an equivalent to Lemma~4.5 in
\cite{BE:2008}, although in the present case $U(k)$ is not unitary. For this 
we require the following matrices,
\be
\label{113}
\mf{D}(\bsy{l}):=\left(\begin{array}{ccc}
         D(\bsy{l}) & 0 & 0\\
         0 & D(\bsy{l}) & 0 \\
         0 & 0 & 0_{\mc{E}_{\ex}}
\end{array}\right), \quad\text{with}\quad 
D(\bsy{l}):=\left(\begin{array}{ccc}
         l_1 &  & 0\\
          & \ddots &  \\
         0 &  & l_{|\mc{E}_{\inte}|}
\end{array}\right).
\ee
\begin{lemma}
\label{71}
Let $\lambda(k)$ be an eigenvalue function of $U(k)$ and assume that there exists
$k_0\in\rz$ such that $\lambda(k_0)=1$. Then $\lambda(k)$ is real differentiable at
$k_0$ and there exists a vector $\bx_0\in\ker(\eins_{\mc{E}}-U(k_0))$ such that
\be
\label{82}
\ui\lambda'\lk k_0\rk=2\left<\bx_0,\frac{L}{L^2+k_0^2}\bx_0\right>_{\kz^{E}}
-\left<\bx_0,\mf{D}(\bsy{l})\bx_0\right>_{\kz^{E}}.
\ee
\end{lemma}
\begin{proof}
A calculation that can be copied verbatim from the proof of Lemma~4.5 in \cite{BE:2008} 
yields
\be
\label{75}
\ba
&\ui\left<\bx(k),\lk\frac{-2L}{L^2+k^2}U(k)+U(k)\mf{D}(\bsy{l})\rk\bx(k)\right>_{\kz^{E}}
+\left<\bx(k),U(k)\bx'(k)\right>_{\kz^{E}}\\
&=\lambda'(k)\left<\bx(k),\bx(k)\right>_{\kz^{E}}
+\lambda(k)\left<\bx(k),\bx'(k)\right>_{\kz^{E}}.
\ea
\ee
Since $U(k)$ is real analytic,  
\be
\label{77}
U(k)=U\lk k_0\rk+O\lk k-k_0\rk, \quad {U(k)}^\ast={U\lk k_0\rk}^\ast+O\lk k-k_0\rk,
\ee
holds for $k\rightarrow k_0$.

From \eqref{lambexp} one obtains that $\lambda(k)=1+O(k-k_0)$, and from \eqref{xexp}
that
\be
\bx'(k) = \sum_{n=1}^l\frac{n}{l} \bx_n(k-k_0)^{(n-l)/l}+O\lk\lk k-k_0\rk^{1/l}\rk,
\ee
where $\bx_n\in\ker(U(k_0)-\lambda(k_0)\eins)$ for $n=0,\dots,l-1$. We hence find that
\be
\label{79}
\ba
&\left<\bx(k),U(k)\bx'(k)\right>_{\kz^{E}}\\
&\hspace{1cm}=\sum_{m=0}^{l-1}\sum_{n=1}^{l}\frac{n}{l}\lk k-k_0\rk^{(m+n-l)/l}
  \left<\bx_{m},\bx_{n}\right>_{\kz^{E}}+O\lk \lk k-k_0\rk^{1/l}\rk,
\ea
\ee
as well as
\be
\label{80}
\ba
&\lambda(k)\left<\bx(k),\bx'(k)\right>_{\kz^{E}}\\
&\hspace{1cm}=\sum_{m=0}^{l-1}\sum_{n=1}^{l-1}\frac{n}{l}\lk k-k_0\rk^{(m+n-l)/l}
   \left<\bx_{m},\bx_{n}\right>_{\kz^E}+O\lk \lk k-k_0\rk^{1/l}\rk.
\ea
\ee
Notice that the first terms on the right-hand side of \eqref{79} and \eqref{80} coincide, 
whereas the error terms, in general, are only of the same order. A similar calculation 
reveals that
\be
\label{81}
\left<\bx(k),U(k)\mf{D}(\bsy{l})\bx(k)\right>_{\kz^{E}}
  =\left<\bx_{0},\mf{D}(\bsy{l})\bx_{0}\right>_{\kz^{E}}+
  O\lk \lk k-k_0\rk^{1/l}\rk.
\ee
Subtracting the first terms on the right-hand side of \eqref{79} and \eqref{80}, 
respectively, from \eqref{75}, one can perform the limit $k\rightarrow k_0$ and 
thus obtains the identity \eqref{82}.
\end{proof}

We now address the question of relating the three multiplicities of the (Laplace-)
eigenvalue zero introduced in Definition~\ref{36}. Regarding the quantity $N$ from
\eqref{38} we need to ensure that $\lambda'(0)\neq 0$ for all eigenvalue functions
$\lambda(k)$ with $\lambda(0)=1$. On the other hand, a characterisation of $\wt{N}$
from \eqref{38a} requires the knowledge of $\ker(\eins-U(0))$. For the latter question
we notice that $U(0)=\mf{S}_0T(0;\bsy{l})=(Q^\perp -Q)\mf{J}_{\mc{E}}$, see
\eqref{9006}, \eqref{3.1} and \eqref{42}. 

Guided by the fact that $\mf{S}_0=Q^\perp -Q$, we now assume that $Q$ is {\it any 
orthogonal projector} in $\kz^E$ and define
\be
\label{8575}
\mf{S}_Q:=Q^{\bot}-Q.
\ee
For this quantity we find the following result.
\begin{lemma}
\label{106a}
Let $\mf{S}_Q$ be as in \eqref{8575}, then
\be
\label{110aa}
\ba
\mf{S}_Q\mf{J}_{\mc{E}}\bw &=-\bw\quad\Leftrightarrow \quad 
        \bw\in\left(\ker Q\cap M_{\asy}\right)\oplus\left((\ker Q)^{\bot}\cap M_{\sy}\right),\\
\mf{S}_Q\mf{J}_{\mc{E}}\bw &=\bw\quad\Leftrightarrow \quad 
        \bw\in\left((\ker Q)^{\bot}\cap M_{\asy}\right)\oplus\left(\ker Q\cap M_{\sy}\right).
\ea
\ee
\end{lemma}
\begin{proof}
It follows from Lemma~\ref{52} that in both cases $\bw\in M$.
An arbitrary vector $\bw\in M$ can be uniquely decomposed as
\be
\label{101a}
\bw=\bw_{\sy}+\bw_{\asy},\quad \bw_{\sy}\in M_{\sy},\quad \bw_{\asy}\in M_{\asy}.
\ee
Decomposing further,
\be
\label{103a}
\ba
\bw_{\sy} &=
\bw_{\sy}^1+\bw_{\sy}^2, \quad \bw_{\sy}^1\in\ker Q,\quad  \bw_{\sy}^2\in\lk\ker Q\rk^{\bot},\\
\bw_{\asy} &=
\bw_{\asy}^1+\bw_{\asy}^2, \quad \bw_{\asy}^1\in\ker Q,\quad  \bw_{\asy}^2\in\lk\ker Q\rk^{\bot},
\ea
\ee
and noting that due to \eqref{JP},
\be
\label{102a}
\mf{S}_Q\mf{J}_{\mc{E}}=\Pbse\Psy+\Pb\Pasy-\Pbse\Pasy-\Pb\Psy,
\ee
we obtain that $\mf{S}_Q\mf{J}_{\mc{E}}\bw=\bw$ reads 
\be
\label{104a}
\ba
&\left(\Pbse\Psy+\Pb\Pasy-\Pbse\Pasy-\Pb\Psy\right) 
   \left(\bw_{\sy}+\bw_{\asy}\right)  \\
                & \hspace{1.5cm}=\bw_{sy}^1+\bw_{sy}^2+\bw_{asy}^1+\bw_{asy}^2,
\ea
\ee
or
\be
\bw_{\sy}^1-\bw_{\sy}^2-\bw_{\asy}^1+\bw_{\asy}^2=\bw_{\sy}^1+\bw_{\sy}^2+\bw_{\asy}^1+\bw_{\asy}^2 .
\ee
Hence, $\bw_{\sy}^2+\bw_{\asy}^1 =\bsy{0}$. Due to the linear independence of $\bw_{\sy}^2$ and 
$\bw_{\asy}^1$ and with a similar reasoning for the eigenvalue $-1$ the claims in the 
Lemma follow.
\end{proof}
\begin{cor}
\label{106b}
It follows immediately that
\be
\ker\lk\eins-U(0)\rk = 
\left((\ker Q)^{\bot}\cap M_{\asy}\right)\oplus\left(\ker Q\cap M_{\sy}\right),
\ee
and hence that
\be
\wt{N} = 
\dim\left((\ker Q)^{\bot}\cap M_{\asy}\right)+\dim\left(\ker Q\cap M_{\sy}\right).
\ee
\end{cor}
For a characterisation of $N$ we need to consider the derivative \eqref{82}
at $k_0=0$, since $N=\wt{N}$ iff $\lambda'(0)\neq 0$ for every eigenvalue
of $U$ with $\lambda(0)=1$. The first term on the right-hand side of \eqref{82} 
has to be treated with some care when $k_0=0$ as $L$ is, in general, not invertible. 
However, any normal map $A:\kz^n\to\kz^n$ has a Moore-Bjerhammer-Penrose 
pseudo-inverse that is defined as follows: Denote by $E_0$ the zero-eigenspace and 
introduce the orthogonal decomposition $\kz^n=E_0\oplus E_0^\perp$. Then 
$A^{-1}_{\mathrm{MBP}}:=0_{E_0}\oplus(A|_{E_0^\perp})^{-1}$, i.e., $A$ is inverted on the
non-zero eigenspaces only. We also require the following matrix,
\be
\label{Gdef}
G(\bsy{l}):=\left(\begin{array}{ccc}
D(\bsy{l})^{-1} & -D(\bsy{l})^{-1} & 0\\
-D(\bsy{l})^{-1} & D(\bsy{l})^{-1} & 0\\
0 & 0 & 0_{\mc{E}_{\ex}}
\end{array}\right),
\ee
where $D(\bsy{l})$ is defined in \eqref{113}.
With \eqref{Psyays} and \eqref{113} one obtains that 
\be
\label{Grel}
G(\bsy{l})=2P_{\asy}\mf{D}(\bsy{l})^{-1}_{\mathrm{MBP}}
=2P_{\asy}\mf{D}(\bsy{l})^{-1}_{\mathrm{MBP}}P_{\asy}.
\ee
\begin{lemma}
\label{232}
The eigenvalues of $L^{-1}_{\mathrm{MBP}}G(\bsy{l})$ are real. Moreover, if 
$l_{\min}>2/\lambda^+_{\min}$, the largest eigenvalue $\tau_{\max}$ of 
$L^{-1}_{\mathrm{MBP}}G(\bsy{l})$ satisfies $\tau_{\max}<1$. 
\end{lemma}
\begin{proof}
We define 
\be
\mf{M}:=2\mf{D}(\bsy{l})^{-1/2}_{\mathrm{MBP}}P_{\asy}L^{-1}_{\mathrm{MBP}}P_{\asy}\mf{D}(\bsy{l})^{-1/2}_{\mathrm{MBP}}, 
\ee
and notice that $\mf{M}$ is self-adjoint. Through a cyclic permutation of the factors, 
$\mf{M}$ can be transformed into $L^{-1}_{\mathrm{MBP}}G(\bsy{l})$, see \eqref{Grel}. Hence the non-zero
spectrum of $\mf{M}$ coincides with that of $L^{-1}_{\mathrm{MBP}}G(\bsy{l})$. Thus the eigenvalues of
$L^{-1}_{\mathrm{MBP}}G(\bsy{l})$ are real.

Let $\tau_{\max}$ be the largest eigenvalue of $\mf{M}$ with associated normalised eigenvector 
$\bsy{x}$, then
\be
\label{new3}
\ba
\tau_{\max}&=\left<\bsy{x},2\mf{D}(\bsy{l})^{-1/2}_{\mathrm{MBP}}P_{\asy}L^{-1}_{\mathrm{MBP}}P_{\asy}
             \mf{D}(\bsy{l})^{-1/2}_{\mathrm{MBP}}\bsy{x}\right>_{\kz^{E}} \\
          &=2\left<P_{\asy}\mf{D}(\bsy{l})^{-1/2}_{\mathrm{MBP}}\bsy{x},L^{-1}_{\mathrm{MBP}}P_{\asy}
             \mf{D}(\bsy{l})^{-1/2}_{\mathrm{MBP}}\bsy{x}\right>_{\kz^{E}}\\ 
          &\leq 2\left\|\mf{D}(\bsy{l})^{-1/2}_{\mathrm{MBP}}\bsy{x}\right\|_{\kz^{E}}^2
             \frac{1}{\lambda^+_{\min}}\\
          &\leq \frac{2}{l_{\min}\lambda^+_{\min}}.
\ea
\ee
Hence, the assumption $l_{\min}>2/\lambda^+_{\min}$ implies $\tau_{\max}<1$.
\end{proof}
We remark that the condition $l_{\min}>2/\lambda^+_{\min}$ is sufficient for $\tau_{\max}<1$,
but not necessary. As an example, choose $L$ such that $L^{-1}_{\mathrm{MBP}}=P_{\sy}$. Then
$L^{-1}_{\mathrm{MBP}}G(\bsy{l})=0$, independent of the choice of edge lengths.
\begin{prop}
\label{NNtildeprop}
Assume that $\tau_{\max}<1$, where $\tau_{\max}$ is the largest 
eigenvalue of $L^{-1}_{\mathrm{MBP}}G(\bsy{l})$. Then $N=\wt{N}$.
\end{prop}
\begin{proof}
When $k_0=0$, the relation \eqref{82} reads
\be
\ui\lambda'\lk 0\rk=2\left<\bx_0,L^{-1}_{\mathrm{MBP}}\bx_0\right>_{\kz^{E}}
-\left<\bx_0,\mf{D}(\bsy{l})\bx_0\right>_{\kz^{E}}.
\ee
Since $\bx_0\in\ker(\eins-U(0))$, the relation \eqref{110aa} together with
$\ran L\subset\ker Q^\perp$ implies that $L^{-1}_{\mathrm{MBP}}\bx_0=P_{\asy}L^{-1}_{\mathrm{MBP}}P_{\asy}\bx_0$.
We define $\by_0:=\mf{D}(\bsy{l})^{1/2}\bx_0$ and notice that, as $\bx_0\in M$, this can
be inverted to yield $\bx_0=\mf{D}(\bsy{l})^{-1/2}_{\mathrm{MBP}}\by_0$. Altogether, this implies
\be
\label{new1}
\ba
\ui\lambda'\lk 0\rk
  &=\left<\by_0,2\mf{D}(\bsy{l})^{-1/2}_{\mathrm{MBP}}P_{\asy}L^{-1}_{\mathrm{MBP}}P_{\asy}
     \mf{D}(\bsy{l})^{-1/2}_{\mathrm{MBP}}\by_0\right>_{\kz^{E}}-\left<\by_0,\by_0\right>_{\kz^{E}}\\
  &=\left<\by_0,\mf{M}\by_0\right>-\left<\by_0,\by_0\right>_{\kz^{E}}.
\ea
\ee
Since by the proof of Lemma~\ref{232} the largest eigenvalue is $\tau_{\max}$,
we conclude that $\lambda'\lk 0\rk\neq 0$.

The secular function can be written as
\be
F(k) = \tilde F(k)\prod_{j=1}^{\wt{N}}\lk 1-\lambda_j(k)\rk,
\ee
where the $\lambda_j(k)$ are the eigenvalue functions of $U(k)$ with $\lambda_j(0)=1$, 
and $\tilde F(0)\neq 0$. Expanding $\lambda_j(k)$ for small $|k|$ then shows that 
$\lambda_j'(0)\neq 0$ implies $N=\ord [F]|_{k=0}=\wt{N}$.
\end{proof}
Under the assumption in the Proposition one can therefore speak of {\it the}
algebraic multiplicity of the Laplace eigenvalue zero.
\section{Zero modes}
\label{111}
In \cite{Kurasov:2010} zero modes were characterised in terms of the quadratic
form associated with the Laplacian $-\Delta_{P,L}$,
\be
\label{112}
\left\|\psi'\right\|^2_{L^2(\Gamma)}-
\left<\underline{\psi},L\underline{\psi}\right>_{\kz^{E}}=0.
\ee
For non-Robin vertex conditions, when $L=0$, it follows that a zero mode must be
constant on each edge, i.e., 
\be
\label{112a}
\psi_e\lk x_e\rk = \alpha_e,\quad \forall e\in\mc{E}.
\ee  
However, when $L\neq0$ this technique cannot be applied since a priori no information 
is available for the boundary term in \eqref{112}.  

Exploiting the fact that $-\Delta\psi=0$, with $\psi\in H^2(\Gamma)$, implies
\be
\label{115}
\psi_e\lk x_e\rk = 
\begin{cases}
\alpha_e+\beta_e x_e,& e\in\mc{E}_{\inte},\\
0,& e\in\mc{E}_{\ex}.
\end{cases}
\ee  
Conditions on the vectors $\bal=(\alpha_e)^T_{e\in\mc{E}_{\inte}}$ and  
$\bbe=(\beta_e)^T_{e\in\mc{E}_{\inte}}$ were derived in \cite{Schrader:2009a}. 
We shall give the same conditions in a slightly different form, on the
vectors $\lk\bal,\bbe,0\rk^T\in M\subset\kz^E$, and for this purpose 
introduce the matrices
\be
\label{114}
C(\bsy{l}):=\left(\begin{array}{ccc}
\eins_{\mc{E}_{\inte}} & 0 & 0 \\
\eins & D(\bsy{l}) & 0 \\
0 & 0 & 0_{\mc{E}_{\ex}}
\end{array}\right)
\quad\text{and}\quad V:=
\left(\begin{array}{ccc}
0_{\mc{E}_{\inte}} & \eins & 0\\
0 & -\eins_{\mc{E}_{\inte}} & 0\\
0 & 0 & 0_{\mc{E}_{\ex}} \\
\end{array}\right).
\ee
\begin{lemma}
\label{119}
The vector $\lk\bal,\bbe,0\rk^T\in M$ contains the coefficients of a zero-mode, iff 
$\bv:=C(\bsy{l})\lk\bal,\bbe,\bsy{0}_{\mc{E}_{\ex}}\rk^T\in M$ satisfies
\be
\label{120}
\ba
\bsy{0}&=P_{\ran L}\left(L^{-1}_{\mathrm{MBP}}G(\bsy{l})-\eins_{\kz^{E}}\right)\bv,\\
\bsy{0}&=P^{\bot}\bv-\bv,\\
\bsy{0}&=QG(\bsy{l})\bv-G(\bsy{l})\bv.
\ea
\ee
\end{lemma}
\begin{proof}
The vertex conditions are
\be
\ba
\bsy{0} 
 &= (P+L) \lk \bal,\bal +D(\bsy{l})\bbe, \bsy{0} \rk^T + P^\perp I
    \lk \bbe,\bbe, \bsy{0} \rk^T\\
 &= P\bv + \lk LC(\bsy{l})+P^\perp V\rk (\bal,\bbe,\bsy{0})^T.
\ea
\ee
Due to orthogonality, the first and the second term in the second line vanish
separately. Hence the second line in \eqref{120} follows. Moreover,
\be
LC(\bsy{l})+P^\perp V = \lk LC(\bsy{l}) +P_{\ran L}V\rk - \lk QV -V\rk,
\ee
where on the right-hand side the terms in the brackets are orthogonal with 
respect to each other. Hence, 
\be
\label{117}
\ba
\bsy{0} &=P_{\ran L} V \lk\bal,\bbe,\bsy{0}_{\mc{E}_{\ex}}\rk^T
          +LC(\bsy{l})\lk\bal,\bbe,\bsy{0}_{\mc{E}_{\ex}}\rk^T,\\
\bsy{0} &=QV\lk\bal,\bbe,\bsy{0}_{\mc{E}_{\ex}}\rk^T 
          -V\lk\bal,\bbe,\bsy{0}_{\mc{E}_{\ex}}\rk^T.
\ea
\ee
On $M$ the matrix $C(\bsy{l})$ is invertible and its $M$-inverse coincides with 
\be
\label{121}
\ba
C(\bsy{l})^{-1}_{\mathrm{MBP}}=
\left(\begin{array}{ccc}
\eins_{\mc{E}_{\inte}} & 0 & 0 \\
 D(\bsy{l})^{-1}  & D(\bsy{l})^{-1} & 0 \\
0 & 0 & 0_{\mc{E}_{\ex}}
\end{array}\right),
\ea
\ee
implying
\be
\label{122}
G(\bsy{l})=-VC(\bsy{l})^{-1}_{\mathrm{MBP}}.
\ee
Moreover,
\be
\label{123}
L^{-1}_{\mathrm{MBP}}L=LL^{-1}_{\mathrm{MBP}}=P_{\ran L},
\ee
so that \eqref{117} implies the first and the third line of \eqref{120}.
\end{proof}
In view of the first equation in \eqref{120} we need to invert 
$L^{-1}_{\mathrm{MBP}}G(\bsy{l})-\eins_{\kz^{E}}$ on $\ran L$. This can be done when, e.g., 
$1$ is not an eigenvalue of $L^{-1}_{\mathrm{MBP}}G(\bsy{l})$. According to Lemma~\ref{232}
this will be the case if one assumes that $l_{\min}>2/\lambda^+_{\min}$, and hence
the largest eigenvalue of $L^{-1}_{\mathrm{MBP}}G(\bsy{l})$ satisfies $\tau_{\max}<1$. In 
such a case $\lk L^{-1}_{\mathrm{MBP}}G(\bsy{l})-\eins_{\kz^{E}}\rk^{-1}$ exists. 

We are now in a position to formulate our characterisation of the zero modes.
\begin{theorem}
\label{272}
Let $-\Delta_{P,L}$ be an arbitrary self-adjoint realisation of the Laplacian on a metric
graph $\Gamma$ and assume that the maximal eigenvalue of $L^{-1}_{\mathrm{MBP}}G(\bsy{l})$ satisfies 
$\tau_{\max}<1$. Then the coefficients $\bal=(\alpha_e)^T_{e\in\mc{E}_{\inte}}$ and 
$\bbe=(\beta_e)_{e\in\mc{E}_{\inte}}^T$ of the zero modes \eqref{115} are of the form
\be
\label{270}
\lk\bal,\bsy{0}_{\mc{E}_{\inte}},\bsy{0}_{\mc{E}_{\ex}}\rk^T \in 
\bma{cc}
\eins_{|\mc{E}_{\inte}|} & 0 \\
 0 & 0_{\mc{E}_{\inte}\oplus\mc{E}_{\ex}} 
\ema
\left(\ker Q \cap M_{\sy}\right).
\ee
In particular, $\bbe=\bsy{0}$ so that all zero modes are edge-wise constant. Furthermore, 
the spectral multiplicity $g_0$ of the Laplace eigenvalue zero is
\be
\label{271}
g_0=\dim\lk \ker Q\cap M_{\sy}\rk.
\ee  
\end{theorem}
\begin{proof}
We recall from Lemma~\ref{119} that 
$\bv=C(\bsy{l})\lk\bal,\bbe,\bsy{0}_{\mc{E}_{\ex}}\rk^T\in M$ has to satisfy three
equations in order for $\bal$ and $\bbe$ to be related to a zero mode. The space 
of solutions $\bsy{v}$ of the first equation is
\be
\label{251}
\mf{V}_1:=\ran\left.\left(L^{-1}_{\mathrm{MBP}}G(\bsy{l})-\eins_{\kz^{E}}\right)^{-1}\right|_{\ker L}.
\ee
The case $\ker L=\emptyset$ implies $\mf{V}_{1}=\emptyset$. Hence we consider 
$\ker L\neq\emptyset$. Under the assumption $\tau_{\max}<1$ the matrix
$\lk L^{-1}_{\mathrm{MBP}}G(\bsy{l})-\eins_{\kz^{E}}\rk$ is invertible, and rearranging its 
Cayley-Hamilton polynomial gives
\be
\label{235}
\ba
\lk L^{-1}_{\mathrm{MBP}}G(\bsy{l})-\eins_{\kz^{E}}\rk^{-1}
 &=\sum\limits_{n=0}^{E-1}a_n\lk L^{-1}_{\mathrm{MBP}}G(\bsy{l})-\eins_{\kz^{E}}\rk^n\\
 &=\sum\limits_{n=0}^{E-1}b_n\lk L^{-1}_{\mathrm{MBP}}G(\bsy{l})\rk^n\\
 &=b_0\eins_{\kz^{E}}+L^{-1}_{\mathrm{MBP}}Z,
\ea
\ee
with complex coefficients $a_n$, $b_n$ and some matrix $Z\in\Mat\lk E\times E;\kz\rk$. 
Choosing $\bsy{w}\in\ker L\neq\emptyset$ one notices that
\be
\label{252}
\bsy{v}=\lk L^{-1}_{\mathrm{MBP}}G(\bsy{l})-\eins_{\kz^{E}}\rk^{-1}\bsy{w}\quad\Leftrightarrow 
\quad\lk L^{-1}_{\mathrm{MBP}}G(\bsy{l})-\eins_{\kz^{E}}\rk\bsy{v}=\bsy{w}
\ee
requires $b_0\neq 0$. The second condition in Lemma~\ref{119}, $P\bv=0$, then 
implies via $\bv=b_0\bsy{w}+L^{-1}_{\mathrm{MBP}}Z\bsy{w}$ that $P\bsy{w}=\bsy{0}$, hence 
$\bsy{w}\in\ker P\cup\ker L=\ker Q$. Therefore, the space of solutions of the first 
and second equation in Lemma~\ref{119} is
\be
\label{250}
\mf{V}_{1,2}=\ran\left.\lk L^{-1}_{\mathrm{MBP}}G(\bsy{l})-\eins_{\kz^{E}}\rk^{-1}\right|_{\ker Q}.
\ee
The third equation adds the condition
\be
\label{3rdcond}
G(\bsy{l})\lk L^{-1}_{\mathrm{MBP}}G(\bsy{l})-\eins_{\kz^{E}}\rk^{-1}\bsy{w}\in(\ker Q)^\perp
\quad\text{with}\quad\bsy{w}\in\ker Q.
\ee
In order to proceed from here we notice that
\be
\label{255}
G(\bsy{l})\lk L^{-1}_{\mathrm{MBP}}G(\bsy{l})-\eins_{\kz^{E}}\rk
=\lk G(\bsy{l})L^{-1}_{\mathrm{MBP}}-\eins_{\kz^{E}}\rk G(\bsy{l}),
\ee
implies
\be
\label{256}
G(\bsy{l})\lk L^{-1}_{\mathrm{MBP}}G(\bsy{l})-\eins_{\kz^{E}}\rk^{-1}
=\lk G(\bsy{l})L^{-1}_{\mathrm{MBP}}-\eins_{\kz^{E}}\rk^{-1}G(\bsy{l}).
\ee
Hence, using that $L^{-1}_{\mathrm{MBP}}$ and $G(\bsy{l})$ are hermitian, it follows 
that \linebreak
$G(\bsy{l})\lk L^{-1}_{\mathrm{MBP}}G(\bsy{l})-\eins_{\kz^{E}}\rk$ is hermitian too.
Furthermore, with \eqref{Gker} and \eqref{256} we find that
\be
\label{259}
G(\bsy{l})\lk L^{-1}_{\mathrm{MBP}}G(\bsy{l})-\eins_{\kz^{E}}\rk^{-1}=
P_{\asy}\lk G(\bsy{l}) L^{-1}_{\mathrm{MBP}}-\eins_{\kz^{E}}\rk^{-1}G(\bsy{l})P_{\asy}.
\ee
We now need an invertible regularisation of $G(\bsy{l})$. For this we
notice from its definition \eqref{Gdef} that $G(\bsy{l})$ is symmetric,
non-negative, and that
\be
\label{Gker}
\ker G(\bsy{l}) = M_{\sy}\oplus M_0 = M_{\asy}^\perp\ .
\ee
This allows us to define
\be
\label{233}
G(\bsy{l})_\delta:=G(\bsy{l}) +\delta^2P_{\asy}^\perp ,\quad \delta>0,
\ee
which is also symmetric and non-negative and, moreover, invertible.
We then obtain 
\be
\label{260}
G(\bsy{l})\lk L^{-1}_{\mathrm{MBP}}G(\bsy{l})-\eins_{\kz^{E}}\rk^{-1}{G(\bsy{l})_\delta}^{-1}=
P_{\asy}\lk G(\bsy{l})L^{-1}_{\mathrm{MBP}}-\eins_{\kz^{E}}\rk^{-1}P_{\asy}.
\ee 
Following an extended version of Sylvester's inertia law \cite{Lancaster:1985}, 
multiplying a hermitian matrix by a positive definite matrix leaves the number 
of positive and negative eigenvalues, respectively, unchanged. Therefore, the 
maps $G(\bsy{l})\lk L^{-1}_{\mathrm{MBP}}G(\bsy{l})-\eins_{\kz^{E}}\rk^{-1}$ and 
$P_{\asy}\lk G(\bsy{l})L^{-1}_{\mathrm{MBP}}-\eins_{\kz^{E}}\rk^{-1}P_{\asy}$ possess the same 
number of positive and negative eigenvalues. The non-zero spectrum of the latter 
is determined by
\be
\label{262}
P_{\asy}\lk G(\bsy{l})L^{-1}_{\mathrm{MBP}}-\eins_{\kz^{E}}\rk^{-1}P_{\asy}\bsy{a}=\mu\bsy{a},
\quad\text{with}\quad \bsy{a}\in M_{\asy}.
\ee
This is equivalent to
\be
\label{265}
\lk G(\bsy{l})L^{-1}_{\mathrm{MBP}}-\eins_{\kz^{E}}\rk \bsy{a}=\frac{1}{\mu}\bsy{a}
\ee
or
\be
G(\bsy{l})L^{-1}_{\mathrm{MBP}}\bsy{a}=\lk L^{-1}_{\mathrm{MBP}}G(\bsy{l})\rk^\ast\bsy{a} =
\lk\frac{1}{\mu}+1\rk\bsy{a}.
\ee
Lemma \ref{232} therefore implies $\mu<0$ so that 
$P_{\asy}\lk G(\bsy{l})L^{-1}_{\mathrm{MBP}}-\eins_{\kz^{E}}\rk^{-1}P_{\asy}$, and by the inertia 
law also $G(\bsy{l})\lk L^{-1}_{\mathrm{MBP}}G(\bsy{l})-\eins_{\kz^{E}}\rk^{-1}$, is negative 
semi-definite. The negative of the latter hence has a hermitian, positive semi-definite 
square-root $\lk G(\bsy{l})\lk\eins_{\kz^{E}}-L^{-1}_{\mathrm{MBP}}G(\bsy{l})\rk^{-1}\rk^{-1/2}$. The 
condition \eqref{3rdcond} together with $\bsy{w}\in\ker Q$ then yields  
\be
\label{269}
\ba
0=&\left|\left<\bsy{w},G(\bsy{l})\lk L^{-1}_{\mathrm{MBP}}G(\bsy{l})-
   \eins_{\kz^{E}}\rk^{-1}\bsy{w}\right>_{\kz^{E}}\right|\\
&\hspace{1.5cm}=\left\|\lk G(\bsy{l})\lk \eins_{\kz^{E}}-L^{-1}_{\mathrm{MBP}}G(\bsy{l})
   \rk^{-1}\rk^{1/2}\bsy{w}\right\|_{\kz^{E}},
\ea
\ee
implying that
\be
\label{3gl}
\bsy{w}\in\ker\lk G(\bsy{l})\lk L^{-1}_{\mathrm{MBP}}G(\bsy{l})-\eins_{\kz^{E}}\rk^{-1}\rk.
\ee
Conversely, every $\bsy{w}$ satisfying \eqref{3gl} also fulfils \eqref{3rdcond}. Hence, 
the conditions \eqref{3gl} and $\bsy{w}\in\ker Q$ are necessary and sufficient for 
$\bsy{w}$ to satisfy \eqref{3rdcond}.

In order to finally prove \eqref{270} we use \eqref{3gl}. From \eqref{262}, the inertia law, 
and the fact that $\mu<0$ one concludes that
\be
\ba
&\ker\lk G(\bsy{l})\lk L^{-1}_{\mathrm{MBP}}G(\bsy{l})-\eins_{\kz^{E}}\rk^{-1}\rk\\ 
&\hspace{0.5cm}= \ker\lk P_{\asy}\lk G(\bsy{l})L^{-1}_{\mathrm{MBP}}-\eins_{\kz^{E}}\rk^{-1}P_{\asy}\rk \\
&\hspace{0.5cm}= M_{\sy}\oplus M_0.
\ea
\ee
The third line follows by expanding $\lk G(\bsy{l})L^{-1}_{\mathrm{MBP}}-\eins_{\kz^{E}}\rk^{-1}$ 
in the same way as in \eqref{235}, showing that this matrix leaves $M_{\asy}$ invariant.

Hence, the conditions \eqref{3gl} and $\bsy{w}\in\ker Q$ are equivalent to
\be
\label{new}
\bsy{w}\in\ker Q\cap\lk M_{\sy}\oplus M_0\rk,
\ee
which, together with the requirement that a zero-mode must vanish on the external edges,
proves \eqref{270}.
\end{proof}
\section{The trace of the $\mf{S}$-matrix and the zero-mode contribution in
the trace formula}
\label{111x}

The form \eqref{115} of zero modes is determined by the fact that apart from
being a solution of $\Delta\psi=0$ one requires $\psi\in L^2(\Gamma)$, which
determines that $\psi_e(x)=0$ when $e\in\mc{E}_{\ex}$. For the following it turns
out that a more general class of zero modes has to be considered.
\begin{defn}
\label{276}
A function $\phi\in C^2(\Gamma)$ is said to be a generalised zero mode, if
$\Delta\phi=0$ and the boundary values $\underline{\phi}$ and $\underline{\phi'}$ 
satisfy the vertex conditions
\be
(P+L)\underline{\phi}+P^\perp I\underline{\phi'}=0.  
\ee
A generalised zero mode $\phi$ is said to be a proper generalised zero mode, if 
$\phi\notin L^2(\Gamma)$.
\end{defn} 
A zero mode is obviously a generalised zero mode too. A proper generalised zero
mode, however, need not be edge-wise constant. A simple example for this is a
single half-line with a Dirichlet vertex condition at the origin. Here the function
$\psi(x)=x$ is a proper generalised zero mode. Edge-wise constant generalised
zero modes will play a particular role.
\begin{defn}
\label{278}
We denote by $\mf{G}_{0}$ the set of edge-wise constant generalised zero modes, and by 
$\mf{G}_{p,0}\subset\mf{G}_{0}$ the subset of edge-wise constant proper generalised 
zero modes. We also set
\be
\label{279}
\wt{g}_{0}:=\dim \mf{G}_0 \quad\text{and}\quad \wt{g}_{p,0}:=\dim \mf{G}_{p,0}.
\ee
\end{defn}
Note that Theorem \ref{272} implies
\be
\label{280}
g_0=\wt{g}_{0}-\wt{g}_{p,0}.
\ee
In order to find expressions for the dimensions \eqref{279} we now associate a compact
metric graph $\wh{\Gamma}$ to a given (in general, non-compact) metric graph $\Gamma$.
This construction is based on cutting every external edge $e\in\mc{E}_{\ex}$ at some
finite value and adding a vertex at the newly created edge end. More precisely, if
$e\in\mc{E}_{\ex}$, we replace this with an edge $\wh{e}$ of length 
$l_{\wh{e}}<\infty$ and add a vertex $v_{\wh{e}}$ that is adjacent only to the edge
$\wh{e}$. We denote the sets of new edges and vertices by $\mc{E}_{\new}$ and 
$\mc{V}_{\new}$, respectively, both of cardinality $|\mc{E}_{\ex}|$. This procedure 
defines a compact graph $\wh{\Gamma}$ with edge set 
$\wh{\mc{E}}=\mc{E}_{\inte}\cup\mc{E}_{\new}$ and vertex set 
$\wh{\mc{V}}=\mc{V}\cup\mc{V}_{\new}$. There is an obvious, isometric embedding 
$\wh{\Gamma}\rightarrow\Gamma$ of the compact graph into the original graph, provided 
by the identity map. This allows us to restrict a function $\psi$ on $\Gamma$ to
$\wh{\Gamma}$, and to identify the restriction $\wh{\psi}:=\left.\psi\right|_{\wh{\Gamma}}$ 
with a function on $\wh{\Gamma}$. The latter has boundary values (if defined), 
\begin{equation}
\label{281}
\underline{\wh{\psi}}:= \left(\left\{\psi_{e}(0)\right\}_{e\in{\mc{E}_{\inte}}},
\left\{\psi_{e}\left(l_{e}\right)\right\}_{e\in{\mc{E}_{\inte}}},\left\{\psi_{e}(0)
\right\}_{e\in{\mc{E}_{\new}}},\left\{\psi_{e}\lk l_e\rk\right\}_{e\in\mc{E}_{\new}}\right)^T,
\end{equation}
in $\kz^{\wh{E}}$ where $\wh{E}=2(|\mc{E}_{\inte}|+|\mc{E}_{\ex}|)$. 

Given a self-adjoint realisation $-\Delta_{P,L}$ of the Laplacian on $\Gamma$, we 
associate with it operators $-\Delta_{\wh{P}_{D/N},\wh{L}}$ on $\wh{\Gamma}$ with 
Dirichlet- and Neumann-vertex conditions at the new vertices $v\in\mc{V}_{\new}$, 
where
\be
\label{282}
\wh{P}_{D/N}:=\left(\begin{array}{cc}
P & 0\\
0 & \eins_{\mc{E}_{\new}}/0_{\mc{E}_{\new}}
\end{array}\right) \quad\text{and}\quad 
\wh{L}:=\left(\begin{array}{cc}
L & 0 \\
0 & 0_{\mc{E}_{\new}}
\end{array}\right).
\ee
Following Theorem~\ref{3000}, the operator $-\Delta_{\wh{P}_{D/N},\wh{L}}$ is 
self-adjoint on the domain corresponding to \eqref{3003a}. The associated 
$\mfS$-matrix is,   
\be
\label{283}
\mf{S}_{\wh{P}_{D/N},\wh{L}}(k)=
\left(\begin{array}{cc}
\mf{S}_{P,L}(k) & 0\\
0 & \mp\eins_{\mc{E}_{\new}}
\end{array}\right),\quad k\notin\ui\sigma(L)\setminus\lgk0\rgk .
\ee
We then set $\wh{Q}_{D/N}:=\wh{P}_{D/N}+P_{\ran\wh{L}}$, as well as
\be
\label{287}
\ba
\wh{M}_{\sy} &:=\lgk\lk\bsy{c}_1,\bsy{c}_1,\bsy{c}_2,\bsy{c}_2\rk^T\in\kz^{\wh{E}}; 
               \quad\bsy{c}_1\in\kz^{|\mc{E}_{\inte}|},\quad\bsy{c}_2\in\kz^{|\mc{E}_{\ex}|}\rgk,\\
\wh{M}_{\asy}&:=\lgk\lk\bsy{c}_1,-\bsy{c}_1,\bsy{c}_2,-\bsy{c}_2\rk^T\in\kz^{\wh{E}};
               \quad\bsy{c}_1\in\kz^{|\mc{E}_{\inte}|},\quad\bsy{c}_2\in\kz^{|\mc{E}_{\ex}|}\rgk.
\ea
\ee
We also define the spectral and algebraic multiplicities $\wh{g_0}_{D/N}$ and 
$\wh{N}_{D/N}$, respectively of the eigenvalue zero of $\Delta_{\wh{P}_{D/N},\wh{L}}$.  
Choosing the shortest additional length $l_e$, $e\in\mc{E}_{\new}$ to be large enough,
the maximal eigenvalue of $\wh{L}^{-1}_{\mathrm{MBP}}G(\wh{\bsy{l}})$ satisfies
$\wh{\tau}_{\max}<1$, if $\tau_{\max}<1$. 
\begin{lemma}
\label{285}
Assume that $\tau_{\max}<1$. Then
\be
\label{284}
\wh{g_0}_D = g_0 ,\quad  \wh{g_0}_N=\wt{g}_0 ,
\ee
and 
\be
\label{286}
\ba
\wh{N}_D&=N+\dim\lk\lk\ker Q\rk^{\bot}\cap\lk M_{\asy}\oplus M_0\rk\rk
          -\dim \lk\lk\ker Q\rk^{\bot}\cap M_{\asy}\rk\\
\wh{N}_N&=\wt{g}_0+\dim\lk(\ker Q)^{\bot}\cap M_{\asy}\rk.
\ea
\ee
\end{lemma}
\begin{proof}
As $\wh{\Gamma}$ has only internal edges, $\psi\in\ker\Delta_{\wh{P}_D,\wh{L}}$ implies 
that $\psi_e(x)=\alpha_e$ for all $e\in\wh{\mc{E}}$. The Dirichlet conditions imposed
at every $v\in\mc{V}_{\new}$ then imply that 
\be
\label{294}
\psi_e\equiv0 \quad \mbox{for all} \quad e\in\mc{E}_{\new}.
\ee
Hence, every zero mode of $-\Delta_{P,L}$ can be identified with a zero mode of 
$-\Delta_{\wh{P}_D,\wh{L}}$, and vice versa. This proves that $\wh{g_0}_D = g_0$. 
Moreover, by \eqref{282},
\be
\label{288}
\ba
\ker\wh{Q}_D &=\lgk\lk\bsy{v},\bsy{0}_{\mc{E}_{\ex}}\rk^T\in\kz^{\wh{E}}; \quad 
               \bsy{v}\in\ker Q\rgk,\\
\lk\ker\wh{Q}_D\rk^{\bot} &=\lgk\lk\bsy{a},\bsy{d}\rk^T\in\kz^{\wh{E}}; \quad 
               \bsy{a}\in(\ker Q)^{\bot}\quad \bsy{d}\in \kz^{|\mc{E}_{\ex}|}\rgk.
\ea
\ee
Using \eqref{287} and \eqref{288} we then obtain
\be
\label{290}
\ba
\dim\lk\ker \wh{Q}_D\cap\widehat{M}_{\sy}\rk=\dim\lk\ker Q\cap M_{\sy}\rk,\\
\dim\lk\lk\ker \wh{Q}_D\rk^{\bot}\cap\wh{M}_{\asy}\rk=\dim\lk(\ker Q)^{\bot}\cap 
\lk M_{\asy}\oplus M_0\rk\rk.
\ea
\ee
Hence, Corollary~\ref{106b} and Proposition~\ref{NNtildeprop} imply the 
first line of \eqref{286}.

Similarly, for $\psi\in\ker\Delta_{\wh{P}_D,\wh{L}}$ the Neumann conditions at
the new vertices, $v\in\mc{V}_{\new}$, impose no restriction on $\psi_e$. Hence,
such a $\psi$ can be extended to ${\Gamma}$ to yield an edge-wise constant 
generalised zero mode, and this process can be reversed. This implies that 
$\wh{g_0}_N=\wt{g}_0$. Furthermore,
\be
\label{402a}
\ba
\ker\wh{Q}_N&=\lgk\lk\bsy{v},\bsy{d}\rk^T\in\kz^{\wh{E}};\quad\bsy{v}\in\ker Q,
              \quad\bsy{d}\in \kz^{E_{\ex}}\rgk,\\
\lk\ker\wh{Q}_N\rk^{\bot}&=\lgk\lk\bsy{a},\bsy{0}_{\mc{E}_\ex}\rk^T\in\kz^{\wh{E}};
              \quad\bsy{a}\in(\ker Q)^{\bot}\rgk.
\ea
\ee
With \eqref{288} and \eqref{402a} this yields
\be
\label{290b}
\ba
\dim\lk\ker\wh{Q}_N\cap\wh{M}_{\sy}\rk=\dim\lk(\ker Q)\cap 
                                      \lk M_{\sy}\oplus M_0\rk\rk,\\
\dim\lk\lk\ker\wh{Q}_N\rk^{\bot}\cap\wh{M}_{\asy}\rk
   =\dim\lk(\ker Q)^{\bot}\cap M_{\asy}\rk.
\ea
\ee
Corollary~\ref{106b}, Proposition~\ref{NNtildeprop} and \eqref{271} 
then imply the second line of \eqref{286}.
\end{proof}
As a consequence, any expression for the number of zero modes in terms of an 
$\mf{S}$-matrix in the sense of \eqref{1021a} will require an expression for 
$\tr\mf{S}_Q$ in terms of the various subspaces of $\kz^E$ that occur above. 
Since \eqref{283} implies
\be
\label{401}
\tr\mf{S}_{\wh{P}_{D/N},\wh{L}}(k)=\tr\mf{S}_{P,L}(k)\mp |\mc{E}_{\ex}|,
\ee
instead of $\mf{S}_Q$ one can work with 
\be
\mf{S}_{\wh{Q}_{D/N}}:=\lim_{k\to 0}\mf{S}_{\wh{P}_{D/N},\wh{L}}(k)
=\wh{Q}_{D/N}^\perp-\wh{Q}_{D/N}. 
\ee
This map, as well as $\mf{J}_{\wh{\mc{E}}}$, is related to the compact graph
$\wh{\Gamma}$. Hence, \eqref{42} implies that $\mf{S}_{\wh{Q}_{D/N}}$ and 
$\mf{J}_{\wh{\mc{E}}}$ are unitary and hermitian in $\kz^{\wh{E}}$ and, therefore,
square to the identity. Thus $\mf{S}_{\wh{Q}}\mf{J}_{\wh{\mc{E}}}$ is unitary, with 
$(\mf{S}_{\wh{Q}}\mf{J}_{\wh{\mc{E}}})^{-1}=\mf{J}_{\wh{\mc{E}}}\mf{S}_{\wh{Q}}$, and 
its eigenvalues lie on the unit circle in $\kz$.
\begin{defn}
\label{820}
Let $\wh{Q}$ be a projector in $\kz^{\wh{E}}$ and set 
$\mf{S}_{\wh{Q}}=\wh{Q}^\perp-\wh{Q}$. We then define the following sets:
\begin{enumerate}
\item The set of all eigenvalues of $\mf{S}_{\wh{Q}}\mf{J}_{\wh{\mc{E}}}$ with positive 
(negative) imaginary part is denoted as $W^{+}$ $(W^{-})$. 
\item Let $\lgk\bz_n\rgk_{1\leq n\leq l}$ be a maximal set of linearly independent 
eigenvectors corresponding to the eigenvalues $\lambda_n\in W^+$. Define   
\be
\label{126}
\mathcal{W}^{(+,\pm)}:=\lin\lgk\lk\mf{J}_{\wh{\mc{E}}}\pm\lambda_n\rk\bz_n; 
\quad 1\leq n\leq l\rgk.
\ee
\item The sets $\mathcal{W}^{(-,\pm)}$ corresponding to $\lambda_n\in W^-$ are defined 
analogously.
\end{enumerate}
\end{defn}
These sets allow us to observe the following.
\begin{lemma}
\label{155a}
Let $\wh{Q}$ be a projector in $\kz^{\wh{E}}$, then
$\bv$ is an eigenvector of $\mfS_{\wh{Q}}$ with eigenvalue $1$, iff  
\be
\label{156a}
\bv\in\mathcal{W}^{(+,+)}\oplus\left(\ker\wh{Q}\cap\wh{M}_{\sy}\right)\oplus
\left(\ker\wh{Q}\cap\wh{M}_{\asy}\right).
\ee
Furthermore, $\va$ is an eigenvector of $\mfS_{\wh{Q}}$ with eigenvalue $-1$, iff  
\be
\label{157a}
\va\in\mathcal{W}^{(+,-)}\oplus\left((\ker\wh{Q})^{\bot}\cap\wh{M}_{\sy}\right)\oplus
\left((\ker\wh{Q})^{\bot}\cap\wh{M}_{\asy}\right).
\ee
Moreover, the relations
\be
\label{127a}
\dim \mathcal{W}^{(+,+)}=\dim \mathcal{W}^{(+,-)}
\ee
and
\be
\label{138a}
\mathcal{W}^{(+,+)}=\mathcal{W}^{(-,+)},\quad \mathcal{W}^{(+,-)}=\mathcal{W}^{(-,-)}
\ee
hold.
\end{lemma}
\begin{proof}
Let $\bz\in\kz^{\wh{E}}$ be an eigenvector of $\mf{S}_{\wh{Q}}\mf{J}_{\wh{\mc{E}}}$ 
with eigenvalue $\lambda$. Due to unitarity, $\lambda^{-1}=\overline{\lambda}$ 
and $\lambda^{-1}\bz=\mf{J}_{\wh{\mc{E}}}\mf{S}_{\wh{Q}}\bz$, so that 
$\mf{S}_{\wh{Q}}\bz=\overline{\lambda}\mf{J}_{\wh{\mc{E}}}\bz$. 
Hence, first,   
\be
\label{8003}
\ba
\mf{S}_{\wh{Q}}\left(\mf{J}_{\wh{\mc{E}}}\bz+\lambda\bz\right)
  &=\mf{J}_{\wh{\mc{E}}}\bz+\lambda\bz,\\
\mf{S}_{\wh{Q}}\left(\mf{J}_{\wh{\mc{E}}}\bz-\lambda\bz\right)
  &=-\lk\mf{J}_{\wh{\mc{E}}}\bz-\lambda\bz\rk,
\ea
\ee
and, furthermore,
\be
\label{8003a}
\mf{S}_{\wh{Q}}\mf{J}_{\wh{\mc{E}}}\mf{J}_{\wh{\mc{E}}}\bz=
\overline{\lambda}\mf{J}_{\wh{\mc{E}}}\bz.
\ee
First consider $\lambda\in W^\pm$, i.e., $\lambda\neq\pm 1$. The eigenspaces 
corresponding to such a $\lambda$ and to $\overline{\lambda}\neq\lambda$,
respectively, are orthogonal. Hence, by \eqref{8003a}, the eigenvectors $\bz$ and 
$\mf{J}_{\wh{\mc{E}}}\bz$ of $\mf{S}_{\wh{Q}}\mf{J}_{\wh{\mc{E}}}$ are orthogonal. Thus 
the vectors $\mf{J}_{\wh{\mc{E}}}\bz\pm\lambda\bz$ are non-zero and, therefore, are 
eigenvectors of $\mf{S}_{\wh{Q}}$ with eigenvalues $\pm 1$, respectively. Hence 
$\mathcal{W}^{(\pm,+)}$ are subspaces of the $\mf{S}_{\wh{Q}}$-eigenspace with eigenvalue 
$1$, and $\mathcal{W}^{(\pm,-)}$ are subspaces of the $\mf{S}_{\wh{Q}}$-eigenspace with 
eigenvalue $-1$, yielding the first direct summands on the right-hand sides of
\eqref{156a} and \eqref{157a}, respectively. Due to the unitarity of 
$\mf{S}_{\wh{Q}}\mf{J}_{\wh{\mc{E}}}$ the cardinalities of $W^+$ and $W^-$ (counted with 
multiplicities) coincide, and also coincide with the dimensions \eqref{127a}.  

Multiplying both equations in \eqref{8003} with $\overline{\lambda}$ we observe 
that the eigenvectors of $\mf{S}_{\wh{Q}}$ corresponding to the eigenvalues $\pm 1$ 
constructed from $\bz$ and from $\mf{J}_{\wh{\mc{E}}}\bz$ are proportional. This 
implies the relations \eqref{138a}.

Now consider $\lambda=\pm 1$, hence Lemma~\ref{106a} applies. For this reason we 
treat $\bz\in \wh{M}_{\sy}$ and $\bz\in \wh{M}_{\asy}$ separately. Notice that \eqref{JP} 
implies $\mf{J}_{\wh{\mc{E}}}\bz=\bz$ for $\bz\in \wh{M}_{\sy}$ and  
$\mf{J}_{\wh{\mc{E}}}\bz=-\bz$ for $\bz\in \wh{M}_{\asy}$.  Such vectors are automatically 
eigenvectors of $\mf{S}_{\wh{Q}}$ with eigenvalues $\pm 1$. In two of the four cases, 
corresponding to $\lambda=\pm 1$ and $\bz\in \wh{M}_{\sy/\asy}$, the eigenvectors 
$\mf{J}_{\wh{\mc{E}}}\bz\pm\lambda\bz$ of $\mf{S}_{\wh{Q}}$ vanish, and in the remaining 
cases they are $\pm 2\bz$. Hence, $\bz$ is an eigenvector of $\mf{S}_{\wh{Q}}$. Therefore, 
the dimensions of the direct sum of the eigenspaces of $\mf{S}_{\wh{Q}}\mf{J}_{\wh{\mc{E}}}$
corresponding to $\lambda=\pm1$ and of the eigenspaces of $\mf{S}_{\wh{Q}}$ for the 
eigenvalues $\pm1$ (via \eb{8003}) are the same. Via Lemma~\ref{106a} this finalises
the proof of the identities given in \eqref{156a} and \eqref{157a}.
\end{proof} 
We are now in a position to determine the trace of $\mf{S}_{\wh{Q}}$.
\begin{prop}
\label{300a}
Let $\wh{Q}$ be an orthogonal projector in $\kz^{\wh{E}}$. Then,
\be
\label{158a}
\ba
\tr\mf{S}_{\wh{Q}} 
 &=2\left[\dim\left(\ker\wh{Q}\cap\wh{M}_{\sy}\right)-\dim\left((\ker\wh{Q})^{\bot}
     \cap\wh{M}_{\asy}\right)\right]\\
 &=2\left[\dim\left(\ker\wh{Q}\cap\wh{M}_{\asy}\right)-\dim\left((\ker\wh{Q})^{\bot}
     \cap\wh{M}_{\sy}\right)\right].
\ea
\ee
\end{prop}
\begin{proof}
We first show the second equality in \eqref{158a}, and assume that
\be
\label{160a}
(\ker\wh{Q})^{\bot}\cap\wh{M}_{\sy/\asy}\neq\{\bsy{0}\}, \quad 
\ker\wh{Q}\cap\wh{M}_{\sy/\asy}\neq\{\bsy{0}\},
\ee
and
\be
\label{161a}
\dim\mathcal{W}^{(+,+)}=\dim\mathcal{W}^{(+,-)}\neq 0.
\ee
All other cases can be treated analogously.
We define
\be
\label{b1}
\ba
M_1&:=\kz^{E}\ominus\lk\lk(\ker\wh{Q})^{\bot}\cap\wh{M}_{\asy}\rk\oplus\lk\ker\wh{Q}
      \cap\wh{M}_{\asy}\rk\rk,\\ 
M_2&:=\kz^{E}\ominus\lk\lk(\ker\wh{Q})^{\bot}\cap\wh{M}_{\sy}\rk\oplus\lk\ker\wh{Q}
      \cap\wh{M}_{\sy}\rk\rk.
\ea
\ee
From Lemma~\ref{155a} we then deduce that
\be
\label{167a}
M_1\oplus M_2=\mathcal{W}^{(+,+)}\oplus\mathcal{W}^{(+,-)}
\ee
holds. Let us assume that 
\be
\label{168a}
M_1\cap\mathcal{W}^{(+,+)}\neq\left\{0\right\},
\ee
and choose $\bsy{0}\neq \bsy{b}\in M_1\cap\mathcal{W}^{(+,+)}$. 
We recall that all elements of $\mathcal{W}^{(+,+)}$ are eigenvectors of 
$\mfS_{\wh{Q}}$ corresponding to the eigenvalue $1$ and therefore are elements of 
$\ker\wh{Q}$. By construction we infer
\be
\label{170a}
\bsy{b}\in\ker\wh{Q}\cap\wh{M}_{\asy},
\ee
which leads to a contradiction to $\bsy{b}\in M_1$. Thus, we obtain
\be
\label{171a}
M_1\cap\mathcal{W}^{(+,+)}=\left\{0\right\}.
\ee
Similar arguments lead to 
\be
\label{172a}
M_1\cap \mathcal{W}^{(+,-)}=\left\{0\right\}, \quad 
M_2\cap \mathcal{W}^{(+,+)}=\left\{0\right\},\quad 
M_2\cap \mathcal{W}^{(+,-)}=\left\{0\right\}.
\ee
Therefore, every element of $\mathcal{W}^{(+,+)}$ and $\mathcal{W}^{(+,-)}$ is an 
orthonormal sum of non-zero vectors in $M_1$ and $M_2$. Eq.\ \eqref{127a} hence implies 
\be
\label{298a}
\ba
&\dim\lk(\ker\wh{Q})^{\bot}\cap\wh{M}_{\asy}\rk+
     \dim\lk\ker\wh{Q}\cap\wh{M}_{\asy}\rk\\
&\hspace{3cm}=\dim\lk(\ker\wh{Q})^{\bot}\cap\wh{M}_{\sy}\rk+
     \dim\lk\ker\wh{Q}\cap\wh{M}_{\sy}\rk.
\ea
\ee
Due to Lemma~\ref{155a} 
\be
\label{8301}
\ba
\tr\mfS_{\wh{Q}}
  &=\dim\lk\ker\wh{Q}\cap\wh{M}_{\sy}\rk+\dim\lk\ker\wh{Q}\cap\wh{M}_{\asy}\rk\\
  &\hspace{1cm}-\dim\lk(\ker\wh{Q})^{\bot}\cap \wh{M}_{\sy}\rk-\dim\lk(\ker\wh{Q})^{\bot}
     \cap \wh{M}_{\asy}\rk.
\ea
\ee
Combining this with \eqref{298a} gives the claim. 
\end{proof}
Now we can present our final result for the trace of the $\mf{S}$-matrix.
\begin{theorem}
\label{300}
%
%
Given an arbitrary self-adjoint realisation $\Delta_{P,L}$ of the Laplacian 
on a metric graph with spectral and algebraic multiplicities $g_0$ and $N$
of the eigenvalue zero, respectively. Furthermore, assume that the maximal 
eigenvalue of $L^{-1}_{\mathrm{MBP}}G(\bsy{l})$ satisfies $\tau_{\max}<1$. Then  
\be
\label{370}
g_0-\frac{N}{2}=\frac{1}{4}\tr\mfS_0 +\frac{|\mc{E}_{\ex}|}{4}-\frac{\wh{g}_{p,0}}{2}
\ee
holds, where $\wh{g}_{p,0}$ is the dimension of the space of edgewise constant
proper generalised zero modes.
\end{theorem}
\begin{proof}
Applying Corollary~\ref{106b} and Theorem~\ref{272} to $-\Delta_{\wh{P}_D,\wh{L}}$ 
gives
\be
\label{406}
2\wh{g_0}_D-\wh{N}_D=\frac{\tr\mf{S}_0\lk\wh{P}_{D},\wh{L}\rk}{2}=
\frac{\tr\mf{S}_0}{2}-\frac{|\mc{E}_{\ex}|}{2},
\ee
where the second equality follows from \eqref{283}.
On the other hand, Lemma \ref{285} yields
\be
\label{407}
\ba
&2\wh{g_0}_D-\wh{N}_D=2g_0-N\\
&\hspace{8mm}-\lk\dim\lk\lk\ker Q\rk^{\bot}\cap\lk M_{\asy}\oplus M_0\rk\rk
   +\dim \lk\lk\ker Q\rk^{\bot}\cap M_{\asy}\rk\rk.
\ea
\ee
In order to calculate the expression in the last line we use Proposition~\ref{300a} 
to obtain
\be
\label{404}
\ba
&\tr\mfS_0\lk\wh{P}_{N},\wh{L}\rk-\tr\mfS_0\lk\wh{P}_{D},\wh{L}\rk\\
&\qquad=2\left[\dim\lk\lk\ker Q\rk^{\bot}\cap\lk M_{\asy}\oplus M_0\rk\rk-
    \dim \lk\lk\ker Q\rk^{\bot}\cap M_{\asy}\rk\right]\\
&\quad\qquad-2\left[\dim\lk\ker Q\cap M_{\sy}\rk-\dim\lk\ker Q\cap\lk M_{\sy}\oplus 
    M_0\rk\rk\right].
\ea
\ee
On the other hand, by \eqref{283} we find that
\be
\label{ds1}
\tr\mfS_0\lk\wh{P}_{N},\wh{L}\rk-\tr\mfS_0\lk\wh{P}_{D},\wh{L}\rk
=2\left|\mc{E}_{\ex}\right|.
\ee
Therefore, using Theorem~\ref{272} together with \eqref{290b} gives
\be
\label{3.14}
\ba
&\dim\lk\lk\ker Q\rk^{\bot}\cap\lk M_{\asy}\oplus M_0\rk\rk-\dim \lk\lk
   \ker Q\rk^{\bot}\cap M_{\asy}\rk\\
&\qquad=\left|\mc{E}_{\ex}\right|+g_0-\wh{g}_{0,N}.
\ea
\ee
Combining \eqref{406}, \eqref{407} and \eqref{3.14} then proves the claim.
\end{proof}
In the case of a compact graph $\Gamma$, where $\mc{E}_{\ex}=\emptyset$
and there are no proper generalised zero modes, the left-hand side of \eqref{370}
is the constant $\gamma$ appearing in the zero-mode contribution to the trace 
formula \eqref{TF}, compare \eqref{gammaS0}.
\begin{cor} 
\label{zerocor}
In the situation of Theorem~\ref{300} assume that
$\Gamma$ is a compact metric graph. Then
\be
\label{370a}
\gamma=g_0-\frac{N}{2}=\frac{1}{4}\tr\mfS_0 .
\ee
\end{cor}
This is an extension of the result in \cite{Fulling:2007} where the case of non-Robin 
vertex conditions, corresponding to $L=0$, was treated. In that case $\mf{S}_{P,L}(k)$
is independent of $k$. 
\section{Factorisation of a non-positive Laplacian: a simple example}
\label{8403}
An interval is the simplest example of a metric graph. For this case we first construct 
a Dirac operator in the sense of \cite{Post:2009} along with a factorisation of the 
associated Laplacian in the form $-\Delta =pp^\ast$, where $p$ is a suitable 'momentum' 
operator. We then modify the Laplacian in order to generate a negative spectrum and 
discuss the necessary modifications in order to preserve the factorisation.

The metric graph $\Gamma$ has one edge that is identified with the interval $\lk0,l\rk$, 
and vertices $\mc{V}=\lgk v_1,v_2\rgk$ at the edge ends. 

In this case the $0$-form space is $L^2\lk\Gamma\rk=L^2\lk0,l\rk$ and the 
$1$-form space is $L^2\lk0,l\rk\oplus\kz^2$. The Dirac operator 
$D$ is defined in the Hilbert space 
\be
\label{Hexample}
\mf{H}_\Gamma = 
L^2\lk0,l\rk\oplus\left(L^2\lk0,l\rk\oplus\kz^2\right).
\ee
We require the associated Laplacian to be local and hence we introduce 
the linear maps
\be
P=
\bma{cc}
P_1 & 0 \\ 0 & P_2
\ema
\qquad\text{and}\qquad
L=
\bma{cc}
L_1 & 0 \\ 0 & L_2 
\ema ,
\ee
on $\kz^2$, where $P_j\in\{0,1\}$. Furthermore, we require $L_j\leq 0$ and, in
particular, $L_j=0$ if $P_j=1$. This ensures that $L$ is defined on 
$\ker P=\ran P^\perp$. The map \eqref{30} takes the form
\be
I=
\bma{cc}
1 & 0 \\ 0 & -1
\ema .
\ee
Next we define the operators  
$d:\mc{D}_{d}\rightarrow L^2\lk0,l\rk\oplus\kz^2$ and $d^\ast:\mc{D}_{d^\ast}\rightarrow L^2\lk0,l\rk$. On the domain
\be
\label{356a}
\mc{D}_{d}:=\left\{\phi\in H^1\lk0,l\rk; \quad P\underline{\phi}=0\right\},
\ee
we set
\be
\label{356c}
d\phi := \lk -\phi',\sqrt{-L}\underline{\phi}\rk.
\ee
Hence, the adjoint operator is defined on
\be
\label{356}
\mc{D}_{d^\ast}:=\left\{(\psi,\bsy{a})\in H^1\lk0,l\rk\oplus\kz^2; \quad 
P^{\bot}I{\underline{\psi}}+\sqrt{-L}\bsy{a}=0\right\}
\ee
to act as
\be
\label{356d}
d^\ast(\psi,\bsy{a}):=\psi'.
\ee
This then implies that the operator $d^\ast d$ is defined on the domain
\be
\label{356b}
\mc{D}_{d^\ast d}=\left\{\psi\in H^2(0,l); \quad P\underline{\psi}=0,\quad
P^\perp I\underline{\psi}'+L\underline{\psi}=0  \right\}
\ee
and acts as
\be
\label{310}
d^{\ast}d\psi =-\psi'' .
\ee
One immediately recognises a self-adjoint realisation of the Laplacian, 
$d^\ast d=-\Delta_{P,L}$, as described in Theorem \ref{3000}.
\begin{rem}
At a vertex $v_j$ where $P_j=1$ the vertex conditions in \eqref{356a}
and \eqref{356b} enforce a Dirichlet condition. When $P_j=0$ and $L_j<0$, however, 
a Robin vertex condition is imposed. This can be seen as a repulsive $\delta$-potential 
of strength $\lambda_j=-2L_j>0$ located at the vertex $v_j$ (cf. 
\cite{KottosSmilansky:1998,Albeverio:2010}).
\end{rem}
Apart from locality, the only restriction imposed in the example above was for $L$
to be negative semi-definite. This was necessary in the factorisation of the Laplacian
due to the square roots in \eqref{356} and \eqref{356c}. Hence, whenever an attractive
$\delta$-potential were to be introduced in a vertex, the associated Laplacian could no
longer be factorised (see also \cite[Section 6]{Fulling:2007}). This observation coincides 
with the fact that a Laplacian with an attractive $\delta$-potential has negative spectrum.

An alternative factorisation of the Laplacian with an attractive $\delta$-potential can 
be guessed from the above example. As one must avoid square roots of $-L$, one could be
tempted to remove the square root in \eqref{356} and, instead, replace the square root 
in \eqref{356c} by a factor $-L$. Formally, \eqref{356b} and \eqref{310} would still 
hold, but the operator $(d^\ast,\mc{D}_{d^\ast})$ would no longer be the adjoint of 
$(d,\mc{D}_{d})$ in the given Hilbert space setting. To remedy this problem one could
modify the inner product in such a way that the strengths $\lambda_j$ of the 
$\delta$-potentials are included in its $\kz^2$-contribution. As the $\lambda_j$'s 
are not necessarily positive, however, this modification would no longer be an 
inner product.

Guided by this observation we conclude that replacing the Hilbert space structure by
a Kre\u{i}n space structure may provide a way to factorise any self-adjoint realisation of
a Laplacian on a graph. Assuming that $L_j\neq 0$, $j=1,2$, in the example above, we
therefore now introduce the non-degenerate hermitian form  
\be
\label{351}
\mf{a}_{L}\lk\lk\phi,\bsy{a}\rk,\lk\psi,\bsy{b}\rk\rk:=
\left<\phi,\psi\right>_{L^2\lk0,l\rk}+\left<\bsy{a},L^{-1}\bsy{b}\right>_{\kz^2}
\ee
on the one-form space $L^2(0,l)\oplus\kz^2$ (viewed as a vector space). This form is used
to define the non-degenerate hermitian form
\be
\label{3781}
\mf{A}_{L}\left((\varphi,(\phi,\bsy{a})),(\xi,(\psi,\bsy{b}))\right)=
\left<\varphi,\xi\right>_{L^2\lk0,l\rk}+\mf{a}_{L}\lk(\phi,\bsy{a}),(\psi,\bsy{b})\rk
\ee
on the vector space 
$\mf{H}_\Gamma=L^2\lk0,l\rk\oplus\left(L^2\lk0,l\rk\oplus\kz^2\right)$ 
(cf.\ \eqref{Hexample}). Equipped with the form \eqref{3781}, this vector space is a
Kre\u{i}n space (see Section \ref{402az} below for more details).

For later purposes we now introduce momentum-like operators instead of the
differential-like quantities \eqref{356d} and \eqref{356c}. On the domain
\be
\label{356aa}
\mc{D}_{p}:=\left\{(\psi,\bsy{a})\in H^1\lk0,l\rk\oplus\kz^2; \quad {P}^{\bot}I
\underline{\psi} -\ui P_{\ran L}\bsy{a}=0\right\}, 
\ee
we define $p:\mc{D}_{p}\to L^2(0,l)$ to act as
\be
p(\psi,\bsy{a}):=-\ui\psi'.
\ee
An adjoint operator $p^\ast:\mc{D}_{p^\ast}\to L^2(0,l)\oplus\kz^2$ with respect to the
hermitian form \eqref{3781} can be defined in close analogy to the Hilbert space
setting in a more or less obvious way (for details see Section~\ref{402az}).
It is defined on the domain
\be
\label{357}
\mc{D}_{p^{\ast}}:=\lgk\phi\in H^1\lk0,l\rk; \quad P\underline{\phi}=0\rgk,
\ee
as
\be
\label{357aa}
p^{\ast}\phi:=\left(-\ui\phi',-L\underline{\phi}\right).
\ee
We now define the Dirac operator
\be
\label{357a}
D:=
\bma{cc}
0 & p \\ p^\ast & 0
\ema
\ee
on the domain
\be
\mc{D}_D = \mc{D}_{p^\ast}\oplus\mc{D}_p
\ee
and observe that this is self-adjoint with respect to the form \eqref{3781}. Moreover,
the operator $pp^\ast$ on the domain
\be
\mc{D}_{pp^\ast}=\{\psi\in H^2(0,l);\quad P^\perp I\underline{\psi'}+(P+L)
\underline{\psi}=0\}
\ee
can be seen to be the self-adjoint realisation of the Laplacian described in 
Theorem \ref{3000}. Note that here $L$ is no longer required to be negative semi-definite.
Below we shall extend this construction to general graphs and general self-adjoint
vertex conditions, and this will include the possibility of a zero eigenvalue
of $L$.
\section{Kre\u{\i}n spaces and momentum operators}
\label{402az}
In order to prepare for the construction of Dirac operators in Kre\u{\i}n spaces we first recall 
some basic facts on Kre\u{\i}n spaces and operators in such spaces, see  \cite{Azizov:1989}
for details.
\begin{defn}
\label{8201}
\begin{itemize}
\item A pair $\lk \mf{K},\left[\cdot,\cdot\right]\rk$, where $\mf{K}$ is a 
complex vector space and $\left[\cdot,\cdot\right]$ is a hermitian form, is 
said to be an indefinite-metric vector space. 
\item Let $(\mf{K},\left[\cdot,\cdot\right])$ be an indefinite-metric vector 
space and let $M\subset \mf{K}$ be a subspace. Then the orthogonal complement 
$M^{\bot}$ is defined as
\be
\label{8206}
M^{\bot}:=\lgk v\in \mf{K}; \quad \left[v,m\right]=0, \quad \forall m\in M\rgk.
\ee
\item A Kre\u{\i}n space is an indefinite-metric vector space admitting a 
$[\cdot,\cdot]$-orthogonal decomposition
\be
\label{8202}
\mf{K}=\mf{K}_+\oplus \mf{K}_-, \quad {\mf{K}_\pm}^{\bot}=\mf{K}_{\mp},
\ee
such that $\lk \mf{K}_+,\left[\cdot,\cdot\right]\rk$ and 
$\lk \mf{K}_-,-\left[\cdot,\cdot\right]\rk$ are Hilbert spaces. 
\item For a fixed decomposition \eqref{8202} we associate to the Kre\u{\i}n space 
$\lk \mf{K},\left[\cdot,\cdot\right]\rk$ a Hilbert space 
$\lk \mf{K},\left(\cdot,\cdot\right)_{\left[\cdot,\cdot\right]}\rk$ by  setting
\be
\label{j1}
(x,y)_{\left[\cdot,\cdot\right]}:=\left[x_+,y_+\right]-\left[x_-,y_-\right],\quad 
\ee 
where $x=x_++x_-$ and $y=y_++y_-$ with $x_\pm,y_\pm\in\mf{K}_\pm$. 
\end{itemize}
\end{defn}
\begin{lemma}[cf. \cite{Azizov:1989}]
\label{8203}
Let $\lk\mf{K}_j;\left[\cdot,\cdot\right]_j\rk$, $j=1,2$, be Kre\u{\i}n spaces and let 
$\mc{O}:\mc{D}_{\mc{O}}\rightarrow \mf{K}_2$ be an operator defined on a dense
domain $\mc{D}_{\mc{O}}\subset\mf{K}_1$. Then there exists a dense subspace
$\mc{D}_{\mc{O}^{\ast}}\subset \mf{K}_2$ and an operator 
$\mc{O}^{\ast}:\mc{D}_{\mc{O}^{\ast}}\rightarrow \mf{K}_1$ such that
\begin{equation}
\label{1211a}
\mc{D}_{\mc{O}^{\ast}} = \{\phi\in\mf{K}_2 ;\quad\exists !\ \varphi\in\mc{\mf{K}}_1\quad
\text{s.t.}\ \left[\phi,\mc{O}\psi\right]_2=\left[\varphi,\psi\right]_1\ \forall\psi\in
\mc{D}_{\mc{O}}\},
\ee
and
\be
\label{1211aaf}
\mc{O}^{\ast}\phi:=\varphi.
\ee
\end{lemma}
\begin{defn}
The operator $\lk\mc{O}^{\ast},\mc{D}_{\mc{O}^{\ast}}\rk$ is said to be the adjoint
operator to $\lk\mc{O},\mc{D}_{\mc{O}}\rk$. An operator is said to be self-adjoint, 
if $\lk\mc{O},\mc{D}_{\mc{O}}\rk=\lk\mc{O}^{\ast},\mc{D}_{\mc{O}^{\ast}}\rk$.
\end{defn}
We now consider an arbitrary metric graph $\Gamma$ and an arbitrary self-adjoint
realisation of a Laplacian on $\Gamma$ as described in Theorem \ref{3000} in 
Section \ref{3010}. The only restriction we impose on the Laplacian is that it is
local in the sense of Definition~\ref{3013}. As, therefore, the self-adjoint map $L$
need not be invertible, the generalisation of the hermitian form \eqref{351}
requires some care. However, we can use the Moore-Bjerhammer-Penrose pseudo-inverse 
$L^{-1}_{\mathrm{MBP}}$ (see Section \ref{3111}). 
\begin{defn}
\label{11}
Let $L$ be self-adjoint on $\kz^E$, satisfying $L=P^\perp L P^\perp$.
\begin{enumerate}
\item The zero-form space $L^2\lk\Gamma\rk$ is equipped 
with a hermitian form given by the natural inner product 
$\left<\cdot,\cdot\right>_{L^2\lk\Gamma\rk}$.
\item The one-form space $L^2\lk\Gamma\rk\oplus\kz^E$ 
is equipped with the hermitian form 
\be
\label{14}
\mf{a}_{L}((\phi,\bsy{a}),(\psi,\bsy{b})):=\left<\phi,\psi\right>_{L^2\lk\Gamma\rk}
+a_L(\bsy{a},\bsy{b}),
\ee
where
\be
\label{z6}
a_L(\bsy{a},\bsy{b}):=\left<\bsy{a},L^{-1}_{\mathrm{MBP}}\bsy{b}\right>_{\kz^E}
\ee
is a hermitian form on $\kz^E$.
\item A hermitian form $\mf{A}_{L}$ on the vector space
$\mf{H}_\Gamma=L^2\lk\Gamma\rk\oplus \lk L^2\lk\Gamma\rk\oplus\kz^E\rk$
is defined as
\be
\label{1920}
\mf{A}_{L}((\varphi,(\phi,\bsy{a})),(\xi,(\psi,\bsy{b}))):=
\left<\varphi,\xi\right>_{L^2\lk\Gamma\rk}+\mf{a}_{L}((\phi,\bsy{a}),(\psi,\bsy{b})).
\ee
\end{enumerate}
\end{defn}
The pairs $\lk L^2\lk\Gamma\rk\oplus\kz^E,\mf{a}_{L}\rk$ and 
$\lk\mf{H}_{\Gamma},\mf{A}_{L}\rk$ are, in general, indefinite-metric spaces. 
In order to identify subspaces $\mf{K}\subset\mf{H}_{\Gamma}$ such that 
$\lk \mf{K},\mf{A}_{L}\rk$ are Kre\u{\i}n spaces, we need some preparations
(see \cite{Azizov:1989}).
\begin{defn}
\label{2001}
\begin{enumerate}
\item A subspace $\mf{E_{+(-)}}\subset\mf{H}_{\Gamma}$ is said to be $\mf{A}_L$-positive
(negative), if
\be
\label{2005}
\mf{A}_{L}(a,a)>0\  (\mf{A}_{L}(a,a)<0)\quad \mbox{for all} \quad a\in\mf{E_{+(-)}}\setminus
\{0\}.
\ee 
\item An $\mf{A}_L$-positive (negative) subspace $\mf{E}_{+(-)}\subset\mf{H}_{\Gamma}$ 
is said to be $\mf{A}_L$-maximally positive (negative) if $\mf{E}_{+(-)}$ is not a proper 
subspace of another $\mf{A}_L$-positive (negative) subspace.
\item $a_L$-positive (negative) subspaces ${E_{+(-)}}\subset\kz^E$ are defined
accordingly.
\end{enumerate}
\end{defn}
\begin{rem}[Example 4.12 in \cite{Azizov:1989}]
\label{2010}
An $\mf{A}_L$-maximally positive (negative) subspace $\mf{E}\subset\mf{H}_{\Gamma}$ 
is not necessarily complete, i.e., it need not be a Hilbert space with respect to 
the norm generated by \eqref{j1}.
\end{rem}
In view of our intention to factorise a Laplacian with core $C^{\infty}_0(\Gamma)$, we
introduce a property of subspaces that will be required later and that is useful to 
identify complete subspaces.
\begin{defn}[$C^{\infty}_0(\Gamma)$-subspace condition]
\label{2350}
Let $\mc{P}_j:\mf{H}_{\Gamma}\rightarrow {L^2\lk\Gamma\rk}$, $j=1,2$, be the
projectors
\be
\label{8016}
\mc{P}_j\lk\phi,(\psi,\bsy{a})\rk=
\begin{cases}
\phi,& j=1,\\
\psi, & j=2,
\end{cases}
\ee
and let $\mc{P}:\mf{H}_{\Gamma}\to L^2\lk\Gamma\rk\oplus L^2\lk\Gamma\rk$ be defined as
\be
\label{8400}
\mc{P}a=\lk\mc{P}_1a,\mc{P}_2a\rk.
\ee
Then a subspace $\mf{E}\subset\mf{H}_{\Gamma}$ is said to be a 
$C^{\infty}_0(\Gamma)$-subspace, if
\be
\label{2723}
{C^{\infty}_0(\Gamma)}\oplus{C^{\infty}_0(\Gamma)}\subset\ran\left.\mc{P}\right|_{\mf{E}}.
\ee
\end{defn}
Let $\mc{M}_{L,+(-)}\subset\kz^E$ be the direct sum of eigenspaces of $L$
corresponding to the positive (negative) eigenvalues. We denote by
$P_{+(-)}:\kz^E\to\kz^E$ the orthogonal projector to $\mc{M}_{L,+(-)}$. 
\begin{lemma}
\label{8108}
Every $\mf{A}_L$-maximally positive $C^{\infty}_0(\Gamma)$-subspace 
$\mf{E}_+\subset\mf{H}_\Gamma$ is closed with respect to $\mf{A}_L$.  
Furthermore, there exists a unique $a_L$-maximally positive subspace 
${E}_{+}\subset\kz^E$ such that
\be
\label{2724}
\mf{E}_+={L^2\lk\Gamma\rk}\oplus
\left({L^2\lk\Gamma\rk}\oplus{E}_{+}\right). 
\ee
In particular, $\lk\mf{E}_+,\mf{A}_L\rk$ is a Hilbert space. Moreover, 
\be
\label{8101a}
P_+{E}_{+}=\mc{M}_{L,+},\quad \dim{E}_{+}=\dim\mc{M}_{L,+}.
\ee
The converse is also true. 
\end{lemma}
\begin{proof}
Let $(\phi,(\psi,\bsy{a}))\in\mf{E}$. As 
$C^{\infty}_0(\Gamma)\subset L^2(\Gamma)$ is dense (in the usual Hilbert
space sense), there exist sequences $\lk\phi_n\rk_{n\in\nz}$ and 
$\lk\psi_n\rk_{n\in\nz}$ in $C^{\infty}_0(\Gamma)$ converging to $\phi$ and $\psi$,
respectively. Thus, for an arbitrary $\epsilon>0$ 
\be
\label{2725a}
-\epsilon+a_L(\bsy{a},\bsy{a})<\mf{A}_L\lk\phi-\phi_n,(\psi-\psi_n,\bsy{a})\rk
<\epsilon+a_L(\bsy{a},\bsy{a}),
\ee
for all $n$ large enough. Thus, 
\be
\label{2725}
{E}_{+}:=\lgk\bsy{a}; \quad (\phi,(\psi,\bsy{a}))\in\mf{E}\rgk
\ee
is positive with respect to $a_L$. Now assume that there exists a maximally positive 
subspace ${E}_{\mf{K}}$ such that ${E}_{+}\subsetneq{E}_{\mf{K}}$. Then 
${L^2\lk\Gamma\rk}\oplus\lk{L^2\lk\Gamma\rk}\oplus{E}_{\mf{K}}\rk$ is 
$\mf{A}_L$-positive with $\mf{E}$ as a proper subspace. This is a contradiction to 
our assumption. Obviously, $L^2\lk\Gamma\rk\oplus L^2\lk\Gamma\rk$ is closed
and complete with respect to $\mf{A}_L$. Since $E_+$ is finite dimensional it is 
also closed and complete. The relations \eqref{8101a} follow directly from 
\cite[Lemma 1.27 and Corollary 1.28]{Azizov:1989}.   
\end{proof}
\begin{rem}
\label{8101}
If $L\not\geq0$ there exist $\mf{A}_L$-maximally positive and complete subspaces
of $\mf{H}_\Gamma$, i.e.,\ Hilbert spaces in the sense of \eqref{j1}, not fulfilling 
the $C^{\infty}_0(\Gamma)$-subspace condition. 
\end{rem}
\begin{proof}
Adapting a construction found in \cite[Example 4.12]{Azizov:1989}, we choose a 
basis $\lgk\mf{e}_n\rgk_{n\in\nz}$ of $L^2\lk\Gamma\rk$ as 
well as $\bsy{e}\in\mc{M}_{L,-}$ with $a_L(\bsy{e},\bsy{e})=-\frac{1}{2}$. 
We then define $\lgk \mf{f}_n\rgk_{n\in\nz}\subset L^2\lk\Gamma\rk\oplus\kz^E$ 
by
\be
\label{l9}
\mf{f}_n:=
\begin{cases}
\mf{e}_1\oplus\bsy{e}, & n=1,\\
\mf{e}_n\oplus\bsy{0}, & n\neq 1,
\end{cases}
\ee
and $F:=\lin\lgk\mf{f}_n\rgk_{n\in\nz}$. We set 
$\mf{F}=L^2\lk\Gamma\rk\oplus F$, where the direct 
sum is taken in $\mf{H}_{\Gamma}$. In analogy to \cite[Example 4.12]{Azizov:1989} 
we infer that $\mf{F}$ is an $\mf{A}_{L}$-maximal positive 
and complete subspace of $\mf{H}_{\Gamma}$. Obviously, there is no $a_L$-positive 
subspace ${E}_{+}$ such that 
$\mf{F}=L^2\lk\Gamma\rk\oplus \lk L^2\lk\Gamma\rk\oplus{E}_{+}\rk$. 
\end{proof}
\begin{lemma}
\label{8251}
Let $\mf{E}_+\subset\mf{H}_\Gamma$ be an $\mf{A}_L$-maximally positive 
$C^{\infty}_0(\Gamma)$-subspace, and let ${E}_{-}\subset\kz^E$ be an 
$a_L$-maximally negative subspace. Then $\lk\mf{K},\mf{A}_L\rk$ is a Kre\u{\i}n space, 
where
\be
\label{8252}
\mf{K}=\mf{K}_+\oplus\mf{K}_-, 
\ee 
with
\be
\label{8252a}
\mf{K}_+={L^2\lk\Gamma\rk}\oplus
\left({L^2\lk\Gamma\rk}\oplus{E}_{+}\right),
\quad \mf{K}_-=\lgk(0,0)\rgk\oplus{E}_{-},
\ee 
and $\mf{A}_L$ is restricted to $\mf{K}$. Moreover,
\be
\label{8253}
P_\pm{E}_{\pm}=\mc{M}_{L,\pm},\quad \dim{E}_{\pm}=\dim\mc{M}_{L,\pm}.
\ee
The converse is also true.
\end{lemma}
\begin{proof}
For the first claim it suffices to notice that $\mf{K}_+$ is a Hilbert space 
by Lemma~\ref{8108}, and $\mf{K}_-$ is obviously a Hilbert space too. 
The relations \eqref{8253} follow from Lemma~\ref{8108} and from
\cite[Lemma 1.27 and Corollary 1.28]{Azizov:1989}. 
\end{proof}
\begin{cor}
\label{8320}
Let ${E}_{+}$ and ${E}_{-}$ be as in Lemma \ref{8251} and set
\be
\label{8321a}
{E}_{\mf{K}}:={E}_{+}\oplus{E}_{-}, \quad \mc{M}_L:=\mc{M}_{L,+}\oplus\mc{M}_{L,-}.
\ee
Then there exists a unique bijective map $P_{\pm}^{-1}:\mc{M}_{L}\rightarrow{E}_{\mf{K}}$
satisfying 
\be
\label{8330}
P_{\pm}^{-1}\lk P_-+P_+\rk=\eins_{{E}_{\mf{K}}}\quad \mbox{and} \quad \lk P_-+P_+
\rk P_{\pm}^{-1}=\eins_{\mc{M}_L}.
\ee
\end{cor}
\begin{proof}
The claim follows directly from \eqref{8101a} and \eqref{8253}. 
\end{proof}
Generalising the constructions in \cite{Fulling:2007,Post:2009}, and guided 
by the example in Section~\ref{8403}, we now define momentum-like operators
and their adjoints. These operators, to which we also simply refer to as momentum 
operators, are then used to define Dirac operators in the same way as in 
\eqref{357a}. Each of these operators is defined in the Kre\u{\i}n space 
${\mf{K}}$ described in Lemma~\ref{8251}. 
\begin{defn}
\label{16}
Let $L$ be self-adjoint on $\kz^E$ such that $P^\perp L P^\perp=L$ and let 
$\lk {\mf{K}},\mf{A}_L\rk$ be the Kre\u{\i}n space 
associated with these data as in Lemma \ref{8251}. Then the generalised 
momentum operator $p_{P,L}:\mc{D}_{p_{P,L}}\rightarrow L^2\lk \Gamma\rk$ 
is defined on the domain
\be
\label{17}
\mc{D}_{p_{P,L}}:=\left\{(\psi,\bsy{a})\in H^1(\Gamma)\oplus E_{\mf{K}};  
\quad P^{\bot}I\underline{\psi}-\ui P_{\ran L}\bsy{a}=\bsy{0}\right\}
\ee
in $L^2(\Gamma)\oplus E_{\mf{K}}$ to act as
\be
\label{17a}
p_{P,L}(\psi,\bsy{a}):=-\ui\psi'. 
\ee 
\end{defn}
\begin{lemma}
\label{19}
The adjoint operator 
$p_{P,L}^\ast:\mc{D}_{p_{P,L}^\ast}\to L^2\lk\Gamma\rk\oplus E_{\mf{K}}$ with respect 
to the hermitian form $\mf{a}_L$ has a domain
\be
\label{21}
\mc{D}_{p_{P,L}^\ast}=\lgk\phi\in H^1(\Gamma); \quad P\underline{\phi}=\bsy{0}\rgk,
\ee
and acts as
\be
\label{21a}
p_{P,L}^\ast\phi= \lk -\ui\phi',-P_{\pm}^{-1}L\underline{\phi}\rk.
\ee
\end{lemma}
\begin{proof}
We first observe that $\mc{D}_{p_{P,L}}$ is dense in 
$\lk\mf{k},\lk\cdot,\cdot\rk_{\mf{a}_L}\rk$, where 
$\mf{k}:=L^2\lk\Gamma\rk\oplus E_{\mf{K}}$ and $\lk\cdot,\cdot\rk_{\mf{a}_L}$ 
is the inner product obtained from $\mf{a}_L$ via \eqref{j1}. Hence, 
Lemma~\ref{8203} implies the existence of an adjoint. Now let 
$\phi\in\mc{D}_{p_{P,L}^{\ast}}$ and $(\psi,\bsy{a})\in\mc{D}_{p_{P,L}}$. An 
integration by parts then shows that
\be
\label{20}
\left<\phi,p_{P,L}(\psi,\bsy{a})\right>_{L^2(\Gamma)}=\ui\left<\underline{\phi},
I\underline{\psi}\right>_{\kz^E}+\left<-\ui\phi',\psi\right>_{L^2(\Gamma)}. 
\ee
Choosing $\psi\in C^{\infty}_0(\Gamma)$ hence implies that $\phi\in H^1(\Gamma)$,
as well as that $\varphi$ in \eqref{1211a} is given by $\varphi=-\ui\phi'$. 
Setting $(\varphi,\bsy{b}):=p_{P,L}^{\ast}\phi$ we obtain
\be
\label{21aa}
\ba
&\mf{a}_{L}(p_{P,L}^{\ast}\phi,(\psi,\bsy{a}))-\left<\phi,p_{P,L}(\psi,\bsy{a})
      \right>_{L^2(\Gamma)}\\
&\hspace{3cm}=\left<\bsy{b},L_{\mathrm{MBP}}^{-1}\bsy{a}\right>_{\kz^E}
      -\ui\left<\underline{\phi},I\underline{\psi}\right>_{\kz^E}.
\ea
\ee
Now, first choosing $\bsy{a}=\bsy{0}$, we conclude from \eqref{17} that 
$I\underline{\psi}\in\ker P^{\bot}$, and then from \eqref{1211a} that
$\underline{\phi}\in\ker P$. Next, let $\left(\psi,\bsy{a}\right)\in\mc{D}_{p_{P,L}}$ 
be arbitrary. From $\underline{\phi}\in\ker P$, the condition in \eqref{17}, 
and the self-adjointness of $L_{\mathrm{MBP}}^{-1}$ we conclude that \eqref{21aa} vanishes, 
iff $\left<L_{\mathrm{MBP}}^{-1}\bsy{b},\bsy{a}\right>_{\kz^E}=-\left<\underline{\phi},
P_{\ran L}\bsy{a}\right>_{\kz^E}$. Hence, 
$P_{\ran L}\underline{\phi}=-L_{\mathrm{MBP}}^{-1}\bsy{b}$ or, equivalently,
$P_{\ran L}\bsy{b}=-L\underline{\phi}$. As 
$-L_{\mathrm{MBP}}^{-1}\bsy{b}\in\mc{M}_L$, Corollary~\ref{8320} implies 
$\bsy{b}=-P_{\pm}^{-1}L\underline{\phi}$. Conversely, every element in \eb{21}
satisfies the corresponding relation \eb{1211a}.  
\end{proof}
\begin{rem}
\label{18}
Definition~\ref{16} and Lemma~\ref{19} are generalisations of \cite{Fulling:2007,Post:2009}
in the following sense:
\begin{itemize}
\item When $L=0$ the space $E_{\mf{K}}$ is trivial and hence can be dropped. The
operators $p$ and $p^\ast$ then are the standard momentum operators defined on 
$H^1(\Gamma)$ with vertex conditions $PI\underline{\psi}=\bsy{0}$ and 
$P\underline{\phi}=0$, respectively. This is the case covered in \cite{Fulling:2007}.
\item  When $L\leq0$ one can replace ${E}_{\mf{K}}$ by $\kz^E$, and $\ui P_{\ran L}$ 
by $\sqrt{-L}$ in \eqref{17}. The action of $p_{P,L}^\ast$ then has to be modified
to $p_{P,L}^\ast (\psi,\bsy{a}):=\lk-\ui\psi',\sqrt{-L}\bsy{a}\rk$. This is the 
construction in \cite{Post:2009} (see also \eqref{356aa}-\eqref{357aa}).
\end{itemize}
\end{rem} 
\section{The Dirac operator and the index theorem}
\label{402}
We now define Dirac operators based on the momentum operators introduced above. 
\begin{defn}
\label{1200}
Let $L$ be self-adjoint on $\kz^E$ such that $P^\perp L P^\perp=L$,
and let $\lk {\mf{K}},\mf{A}_L\rk$ be the Kre\u{\i}n space associated with these data as in
Lemma \ref{8251}. Then a Dirac operator $D_{P,L}$ on the metric graph 
$\Gamma$ is defined on the domain
\be
\label{8322}
\mc{D}_{D_{P,L}}:=\mc{D}_{p_{P,L}^{\ast}}\oplus\mc{D}_{p_{P,L}}
\ee
to act as
\be
D_{P,L}\lk\phi,\lk\psi,\bsy{a}\rk\rk=
\lk p_{P,L}\lk\psi,\bsy{a}\rk,p_{P,L}^\ast\phi\rk,
\ee
i.e., with respect to the decomposition 
$L^2(\Gamma)\oplus\lk L^2(\Gamma)\oplus E_{\mf{K}}\rk$,
\be
D_{P,L}=\begin{pmatrix}
0 & p_{P,L} \\ p^\ast_{P,L} & 0
\end{pmatrix}.
\ee
\end{defn}
\begin{lemma}
\label{8312}
The Dirac operator $\lk D_{P,L},\mc{D}_{D_{P,L}}\rk$ is self-adjoint in 
$\lk \mf{K},\mf{A}_L\rk$ and is closed in $\lk\mf{K},\lk\cdot,\cdot\rk_{\mf{A}_L}\rk$, 
where $\lk\cdot,\cdot\rk_{\mf{A}_L}$ is the inner product corresponding to \eqref{j1}.
\end{lemma}
\begin{proof}
For the self-adjointness of the Dirac operator we have to show that
\be
\label{2025}
p_{P,L}^{\ast\ast}=p_{P,L}. 
\ee
As $\mc{D}_{p_{P,L}^{\ast\ast}}$ contains 
$C^{\infty}_{0}(\Gamma)\oplus C^{\infty}_{0}(\Gamma)\oplus\lgk\bsy{0}\rgk$, an 
integration by parts shows that
\be
\label{2025a}
p_{P,L}^{\ast\,\ast}\lk\psi,\bsy{0}\rk=-\ui\psi' \quad \mbox{for all} \quad 
\psi\in C^{\infty}_{0}(\Gamma).
\ee
Thus, the equivalent of relation \eqref{21aa} holds. Hence, by  
$P\underline{\phi}=\bsy{0}$ for all $\phi\in\mc{D}_{p_{P,L}^{\ast}}$ we infer 
$\psi\in\mc{D}_{p_{P,L}}$. The completeness of 
${H^1(\Gamma)}\oplus{H^1(\Gamma)}\oplus\kz^E$ with respect to
$\left<\cdot,\cdot\right>_{\mf{H}_{\Gamma}}$ as well as the continuity of the 
trace map $\psi\mapsto\underline{\psi}$ imply that $D_{P,L}$ is closed in 
$\lk\mf{K},\lk\cdot,\cdot\rk_{\mf{A}_L}\rk$.
\end{proof}
\begin{rem}
\label{j10}
If one defined a momentum operator $p_{P,L}$ as in \eqref{17}, but with $\kz^E$ 
instead of $E_{\mf{K}}$, one would be closer to the approach of \cite{Post:2009}. 
However, in doing so the adjoint to the momentum operator would have to be a 
multi-valued operator (cf. \cite{Brezis:1973,Bruning:2008}). This in turn would 
generate a multi-valued Dirac operator. In order to obtain a single-valued Dirac 
operator one would have to restrict the multi-valued operator to a Kre\u{\i}n 
space $\mf{K}$. This would yield the same Dirac operator as in Definition~\ref{1200}. 
\end{rem}

For the following result compare Theorem~\ref{3000} in Section \ref{3010}.
\begin{prop}
\label{2700}
Let $D_{P,L}$ be a Dirac operator on a metric graph $\Gamma$. Then its square,
$D_{P,L}^2$, has a domain
\be
\mc{D}_{D_{P,L}^2}=\mc{D}_{\Delta_1}\oplus\mc{D}_{\Delta_2},
\ee
where
\be
\label{2701}
\ba
&\mc{D}_{\Delta_1}:=\lgk\phi\in H^2(\Gamma), \quad \lk P+L\rk\underline{\phi}
      +P^{\bot}I\underline{\phi'}=\bsy{0}\rgk,\\
&\mc{D}_{\Delta_2}:=\lgk\lk\psi,\bsy{a}\rk\in H^2(\Gamma)\oplus{E}_{\mf{K}}; 
      \ P^{\bot}I\underline{\psi}-\ui P_{\ran L}\bsy{a}=\bsy{0}, 
      \quad P\underline{\psi'}=\bsy{0}\rgk,\\
\ea
\ee
on which it acts as
\be
\label{2702}
\ba
D_{P,L}^2\lk\phi,(\psi,\bsy{a})\rk
&=\lk-\Delta_1\phi,-\Delta_2\lk\psi,\bsy{a}\rk\rk\\
&=\lk-\phi'',\lk-\psi'',\ui P_{\pm}^{-1}L\underline{\psi'}\rk\rk.
\ea
\ee
\end{prop}
\begin{proof}
Obviously, $-\Delta_1=p_{P,L}p_{P,L}^{\ast}$ and $-\Delta_2=p_{P,L}^{\ast}p_{P,L}$. 
From \eqref{8330} we conclude $P_{\ran L}P_{\pm}^{-1}L\underline{\phi}=L\underline{\phi}$. 
Definition \ref{16} and Lemma \ref{19} then imply the statements for 
$\mc{D}_{\Delta_1}$, $\mc{D}_{\Delta_2}$ and for the action of $D_{P,L}^2$.
\end{proof}
Applying Theorem \ref{3000} in Section \ref{3010} to Proposition \ref{2700}, 
we obtain the following corollary.
\begin{cor}
\label{31}
$\lk-\Delta,\mc{D}_{\Delta}\rk$ is a self-adjoint quantum graph Laplacian, iff 
there exists a Dirac operator $D_{P,L}$ satisfying 
$\lk\mc{P}_1 D_{P,L}^2\mc{P}_1,\mc{D}_{\Delta_1}\rk=\lk-\Delta,\mc{D}_{\Delta}\rk$.
\end{cor}
Our ultimate goal is to prove an index theorem that is a suitable generalisation 
of \eqref{1021a}. First, however, we express the index of a Dirac operator on a
general metric graph (including non-compact ones) in terms of the subspaces 
defined above. For compact graphs the final version of the index theorem will be 
given below. 
\begin{prop}[Pre-index theorem]
\label{27}
Let $\Gamma$ be an arbitrary metric graph, and let $D_{P,L}$ be an arbitrary 
Dirac operator as defined in Definition~\ref{1200}. Then the kernels of $p_{P,L}$ 
and $p_{P,L}^\ast$ are
\begin{equation}
\label{28}
\ba
\ker p_{P,L}&=\left\{(\psi,\bsy{a})\in\mc{D}_{p_{P,L}};\quad I\underline{\psi}\in
     \ran Q\cap M_{\asy}\right\},\\
\ker{p_{P,L}^\ast}&=\left\{\phi\in H^1(\Gamma); \quad \underline{\phi}\in (\ran Q)^{\bot}
     \cap M_{\sy}\right\},
\ea
\end{equation}
where $Q=P+P_{\ran L}$. Both kernels are finite dimensional, and 
\begin{equation}
\label{29a}
\Ind\lk D_{P,L}\rk=\dim\lk\ker Q\cap M_{\sy}\rk-\dim\lk(\ker Q)^{\bot}\cap M_{\asy}\rk .
\end{equation}
\end{prop}              
\begin{proof}
From \eqref{17} and \eqref{17a} we conclude that if $(\psi,\bsy{a})\in\ker p_{P,L}$ 
then $\psi$ is constant on every edge and vanishes on external edges. Thus, 
$I\underline{\psi}\in M_{\asy}$. From Corollary \ref{8320} we infer that there exists 
exactly one $\bsy{a}\in E_{\mf{K}}$ satisfying 
$P^{\bot}I\underline{\psi}+\ui P_{\ran L}\bsy{a}=\bsy{0}$. Hence, $I\underline{\psi}\in\ran Q$. 
This implies $I\underline{\psi}\in M_{\asy}\cap \ran Q$.  

From \eqref{21} and \eqref{21aa} it follows that if $\phi\in\ker p_{P,L}^\ast$ then 
$\phi$ is constant on every edge and vanishes on external edges. Hence, 
$\underline{\phi}\in M_{\sy}$. Lemma \ref{8312} and Corollary \ref{8320} imply $
P\underline{\phi}=\bsy{0}$ and $L\underline{\phi}=\bsy{0}$. Since $PL=0$ we conclude 
that $\underline{\phi}\in \ran Q^{\perp}$. This yields 
$\underline{\phi}\in \ran Q^{\bot}\cap M_{\sy}$. The rest is obvious.
\end{proof}
In the case of a compact graph, a comparison of \eqref{29a} with 
Proposition~\ref{300a} immediately gives an index theorem.
\begin{theorem}[Index theorem for the compact case]
\label{9001}
Assume that $\Gamma$ is a compact metric graph and let $D_{P,L}$ be a Dirac operator 
as given in Definition~\ref{1200}, with associated Laplacian
$-\Delta_{P,L}=\mc{P}_1 D_{P,L}^2\mc{P}_1$. Then
\be
\label{9007}
\Ind\lk D_{P,L}\rk=\frac{1}{2}\tr\mf{S}_0.
\ee
\end{theorem}
We remark that comparing this result, under the additional assumption $\tau_{\max}<1$,
with Corollary~\ref{zerocor} gives
\be
g_0-\frac{N}{2} = \frac{1}{2}\Ind\lk D_{P,L}\rk,
\ee
which is the zero-mode contribution $\gamma$ to the trace formula
\eqref{TF} (where the stronger assumption $l_{\min}>2/\lambda^+_{\min}$ was made).

\subsection*{Acknowledgment}
S.E.\ thanks P.\ Exner and O.\ Post for stimulating discussions and is very grateful 
to W.\ Arendt for his inspiring advanced lecture course {\it Funktionalanalysis 2} given 
at Ulm University in 2009/10. This research was supported by the German Academic Exchange 
Service (DAAD) as well as by the EPSRC network {\it Analysis on Graphs} (EP/1038217/1).
\bibliographystyle{unsrt}
\bibliography{litver}
\end{document}